\documentclass[journal]{IEEEtran}
\usepackage[utf8]{inputenc}
\usepackage[T1]{fontenc}
\usepackage{amsmath,amssymb,amsfonts}
\usepackage{graphicx}
\usepackage{booktabs}
\usepackage{array}
\usepackage{threeparttable}
\makeatletter
\newcommand{\TPTnoteswidth}[1]{\xdef\TPT@hsize{\hsize#1\relax\noexpand\@parboxrestore}\TPT@hsize}
\makeatother
\usepackage{algorithm}
\usepackage{algpseudocode}
\AtBeginDocument{\setlength{\abovedisplayskip}{4pt plus 1pt minus 2pt}\setlength{\belowdisplayskip}{4pt plus 1pt minus 2pt}\setlength{\abovedisplayshortskip}{2pt}\setlength{\belowdisplayshortskip}{2pt}}
\newtheorem{proposition}{Proposition}

\newenvironment{proof}{\noindent\emph{Proof.} }{\hfill$\blacksquare$\par}
\title{The Operable Pareto Front: Distilling Offline Search into Run-Time Control for Multi-Objective UAV Edge-Computing Scheduling}
\author{
Qiao Liao,
Zhiyong Feng,
Bin Wu,
Guodong Fan

\thanks{Q. Liao, B. Wu are with School of Computer Science and Technology, Tianjin University, Tianjin, China. E-mail:\{liaoqiao,binw\}@tju.edu.cn\\
Z. Feng is with School of Computer Software, Tianjin University, Tianjin, China. E-mail: zyfeng@tju.edu.cn\\
G. Fan is with School of Information Science and Engineering, Shandong Agriculture and Engineering University, Shandong, China. E-mail: guodongfan@tju.edu.cn\\
\vspace{-2mm}

Zhiyong Feng is the corresponding author.}}
\begin{document}
\maketitle

% ==================== ABSTRACT ====================
\begin{abstract}
A UAV mobile edge computing (MEC) fleet trades energy against delay, and its schedules form a Pareto front; we call a scheduler \emph{operable} when the fleet can be asked for any point on that front at run time. We propose PrefDT, to the best of our knowledge the first preference-conditioned Decision Transformer for the problem of joint trajectory, association and offloading scheduling. Its idea comes from language modeling: we hand the model the desired trade-off as an input, such that a single model only needs to be trained once offline to return any desired point on the curve in one rollout. The network state is summarized by attention pooling with a per-user bypass, so the scheduler keeps working when user reports are lost. The energy target is a running budget decremented by what the fleet actually spends. As a result, when wind or load pushes consumption off the plan, the policy can track the difference and hold its budget. Because no corpus of preference-labeled flights exists, we design a distillation pipeline and build the corpus by ourselves. In simulation against 26 method variants, PrefDT produces the best trade-off curve of any learned method and holds its energy budget to within 0.6\% when propulsion cost rises by half in mid-flight.
\end{abstract}


\IEEEpeerreviewmaketitle

% ==================== I. INTRODUCTION ====================
\section{Introduction}
\label{sec:intro}

\IEEEPARstart{U}{nmanned} Aerial Vehicles (UAVs) carrying edge servers extend mobile edge computing (MEC) to places where ground infrastructure is absent, damaged, or overloaded. A fleet is expected to balance two costs: the time users wait and the energy the aircraft burn, and the balance it should strike is not fixed for the mission. In a disaster response, delay matters most while survivors are being located, but energy matters most once the fleet must stay aloft until relief arrives; the same fleet may be asked for two different balances in one day. Generally, the schedule is computed for one balance before takeoff, and a different balance costs another optimization or training run. The schedule is also fixed once the fleet is in the air. It is built from predicted energy consumption and link quality, and cannot notice the difference when wind raises what a maneuver costs, or reports from users fail to arrive. Deployment therefore asks for a scheduler that takes the balance as a request at run time, holds an energy budget against what the fleet actually spends, and keeps deciding when reports are lost.

Solving this problem means planning three quantities at the same time. In every time slot the scheduler decides where each aircraft flies, which users it serves, and how much of each user's task it takes on. These decisions are tightly coupled with each other and cannot be made separately~\cite{wu2018multiuav,jtoratc}. The decisions also serve two goals set against each other: finishing tasks quickly costs energy, and conserving energy costs delay. As a result, there is no single best schedule, and one must instead choose a single point on a trade-off curve called the Pareto front. The scheduling problem is NP-hard.

Three categories of schedulers dominate the current literature: model-based optimization, learned schedulers and population methods. They differ from each other mainly in when they commit to a trade-off. In particular, model-based optimization commits at solve time: based on a chosen weight, the two conflicting goals are collapsed into one, and the resulting program is solved by block-coordinate descent with successive convex approximation (SCA) over detailed channel and propulsion models~\cite{wu2018multiuav,zeng2019uavproc,jtoratc}. As a result, any weight change needs another run, and collapsing the goals into a single optimization objective puts part of the Pareto front permanently out of reach~\cite{das1997drawbacks}. As the second category, learned schedulers commit at training time: a deep offloading policy such as DROO~\cite{huang2020droo}, or a graph-attention scheduler~\cite{feng2024gat}, is trained for one weighting and one scenario. In contrast, population methods (i.e., the third category) postpone the commitment by returning many solutions in a single run. For example, PGMORL~\cite{xu2020pgmorl} evolves a population of policies, while the classical dominance-based evolutionary search NSGA-II~\cite{deb2002nsga2} evolves the solutions themselves. Unfortunately, this category of schedulers generally converges to a fixed menu of configurations rather than a continuum. To obtain a trade-off that is not on the menu, the search process must be run again.

In current literature, all three categories above are evaluated by how good a Pareto front they can produce. Such a criterion says nothing about the cost of changing to a different trade-off, and what happens when a flight does not go as predicted. Moreover, the schedule is computed from predicted consumption, and none of these schedulers tracks the energy already spent; nothing is corrected when the actual flying cost departs from the predicted one.

A conditioned scheduler addresses the first issue above (i.e., the cost of changing to a different trade-off). It takes two inputs: the state of the network, and the energy--delay trade-off the fleet is asked to strike. Three requirements are desirable for a UAV fleet: \textbf{(A)}~When the energy--delay trade-off is set differently, the fleet's behavior follows in proportion, across the whole range the physics permits, and without becoming unstable when more is asked for than can be delivered. \textbf{(B)}~Finish the mission within a given total energy budget, measured against the energy actually consumed rather than a predicted plan---even when conditions change mid-flight (e.g., a headwind or a heavier task load). \textbf{(C)}~Accept a mid-flight revision of that budget, issued as a single new total, without re-planning or retraining. A scheduler is said to be \emph{operable} if it meets all three requirements above, and Sec.~\ref{sec:dissoc} makes the term measurable.

This paper builds an operable scheduler without sacrificing the quality of the Pareto front. The idea comes from language modeling. In particular, Decision Transformer~\cite{chen2021dt} and Trajectory Transformer~\cite{janner2021tt} recast control as conditional sequence generation. A trajectory becomes a stream of tokens; a causal model learns to continue that stream. The decisive part is where the goal goes: the desired outcome is supplied as an input token, which can be set at inference time, instead of being compiled into a training objective. The idea has since been extended in several directions---online adaptation~\cite{zheng2022odt}, prompting~\cite{xu2022promptdt}, and multi-agent execution~\cite{meng2023madt}---and two lines of work carry it toward our setting: conditioning on more than one objective, and use as a resource manager in wireless networks. PEDA~\cite{zhu2023peda} and related multi-objective variants~\cite{ghanem2023modt} condition on a preference over vector-valued returns, evaluated on locomotion benchmarks. A value-based line conditions a critic on the preference rather than a sequence~\cite{yang2019envelope,basaklar2022pdmorl,kostrikov2022iql}. In wireless networking, Decision Transformers have begun to appear as resource managers~\cite{apdt}, though so far with a scalar cost signal and a single objective. If what a policy should achieve can be supplied as an input, then needing a fresh solve for every trade-off follows from where the trade-off enters the pipeline, not from the problem itself.

Directly borrowing the idea from language modeling does not match the physical network, and three gaps exist. The first is the input. The set of users present changes from slot to slot, and each user needs a decision of its own, while a sequence model wants a fixed layout of inputs. Pooling all users into one summary accepts any number of them but erases their identity, which per-user decisions cannot afford. Our encoder therefore keeps the pooled summary for what the fleet has in common and adds a bypass: each user's raw token is routed around the summary and directly into that user's own association and offloading heads. One frozen model can then serve any number of users and still decide for each of them separately. An additional benefit comes for free: the summary is defined on any subset of the users, so a lost report leaves a smaller set rather than an incomplete input, and the scheduler keeps working on a link that drops reports.

The second gap is the setting. A Decision Transformer is steered by a single scalar, the return-to-go; folding delay and energy into one number leaves no part of it in joules. We give one value per objective instead, one for delay and the other for energy. The energy entry starts at the mission's budget and is decremented, at each slot, by what the fleet actually consumed. By arithmetic, that entry is then the remaining allowance itself, and the policy reads it before every decision. The budget is thereby enforced against real consumption rather than against a predicted plan. A mid-flight revision is then a single overwrite of that entry, and, because each objective has an entry of its own, delay and energy targets can be set independently.

The third gap is the training data. This kind of model learns by imitation from recorded runs labeled with the outcome achieved, but no such records exist for our problem. We produce them ourselves: a compact control rule is searched once, and its archive of solutions spans the front; each sampled preference is matched to the archive member whose realized delay--energy share is nearest, and that member's flight is recorded under the corresponding label. The label is therefore a realized trade-off rather than a weight, so a requested setting denotes an outcome the fleet can actually produce. The front's extreme members are demonstrated more often than uniform sampling would, because imitation regresses toward the behavioral mean; the whole corpus costs one offline search, with no per-preference training.

Addressing all three gaps leads to PrefDT, the scheduler proposed in this paper. PrefDT turns the Pareto front from something a search has to produce into something a model can be asked for: one frozen model, trained entirely offline from preference-annotated logs and never updated again, returns any trade-off at the cost of a single rollout.

Our algorithmic contributions are summarized below.

(1) We propose PrefDT, an operable scheduler: one frozen model returns any point on the Pareto front at run time, so a new trade-off costs one rollout instead of another solve, search, or training run. It holds an energy budget against actual consumption and accepts a revised budget in flight.

(2) We bring the Decision Transformer of language modeling to UAV scheduling and adapt it to a physical network. The desired trade-off is supplied as an input at inference time rather than compiled into the training objective, so the scheduler takes its trade-off at run time. We add a per-user bypass to the pooled state encoder, so the scheduler decides for each user separately while the pooled summary accepts any number of users. We replace the scalar goal input with a vector, one value per objective, so every point on the front can be named and energy has a channel of its own. No preference-labeled flight data exists, so we build a distillation pipeline to produce it: a searched control rule supplies the flights, and each flight is labeled with the trade-off it actually struck.

(3) Three mechanisms carry these designs into deployment. Lost reports are tolerated because the pooled summary is defined on any subset of the active users: a lost report leaves a smaller set, and nothing is retrained. Energy is tracked because the energy entry of the goal vector starts at the mission's budget and is decremented every slot by what was actually consumed: the policy reads the remaining budget, and a revision is one overwrite of that entry. The model learns the whole front including extremes, because the distillation pipeline matches each sampled preference to the archived rule whose realized delay--energy share is nearest and oversamples the extreme members.

Across 26 method variants under one protocol, PrefDT produces the best trade-off curve of any learned method, and a new trade-off costs one 0.41-second rollout, 620--2400$\times$ less than a re-solve or a retraining. When propulsion cost rises by half in mid-flight, the same frozen model stays within 0.6\% of its energy budget.

The paper is organized as follows. Section~\ref{sec:model} formulates the system model and the scheduling problem. Section~\ref{sec:prefdt} presents PrefDT: the encoder, the conditioning channel, and the distillation pipeline. Section~\ref{sec:setup} describes the experimental setup and the operability criterion, and Section~\ref{sec:results} reports the experiments. Section~\ref{sec:conclusion} concludes the paper and states the limitations of this work. Appendices~A--D are provided in the supplementary material.

% ==================== III. SYSTEM MODEL ====================
\section{System Model and Problem Formulation}
\label{sec:model}
Physical models follow the multi-UAV MEC literature~\cite{zeng2019uavproc,jtoratc,alhourani2014lap} so that all baselines share an identical environment; Appendix~A collects the symbols.

\subsection{Network model and decision variables}
$M$ rotary-wing UAVs, indexed by $m\in\mathcal{M}$, fly at a fixed altitude $H$ above a square service area and act as edge servers for $U$ ground users during $T$ slots of length $\delta_t$; $\mathbf{q}_m(t)\in\mathbb{R}^2$ denotes the horizontal projection of UAV $m$ and $\mathbf{w}_u\in\mathbb{R}^2$ the (static) position of user $u$. In each slot an active set $\mathcal{K}(t)$ requests service, its size drawn uniformly from $\{4,\dots,8\}$ and its members uniformly at random from the user pool, and each member carries a task with probability $p_a$. The active set therefore has \emph{time-varying} cardinality. UAV kinematics follow
\begin{equation}
\mathbf{q}_m(t{+}1)=\mathbf{q}_m(t)+d_m(t)\big[\cos\varphi_m(t),\ \sin\varphi_m(t)\big]^{\!\top},
\label{eq:kin}
\end{equation}
with step length $0\le d_m(t)\le D_{\max}=V_{\max}\delta_t$, heading $\varphi_m(t)\in[0,2\pi)$, and speed $v_m(t)=d_m(t)/\delta_t$. In every slot the scheduler jointly selects
\begin{equation}
\mathbf{a}_t=\big\{d_m(t),\varphi_m(t)\big\}_{m\in\mathcal{M}}\,\cup\,\big\{x_{u,m}(t)\big\}\,\cup\,\big\{\rho_u(t)\big\}_{u\in\mathcal{K}(t)},
\label{eq:action}
\end{equation}
i.e., continuous flight controls, the binary association $x_{u,m}(t)\in\{0,1\}$ (each active user served by at most one UAV), and the continuous fraction $\rho_u(t)\in[0,1]$ of task $u$ offloaded to its serving UAV. Edge compute is shared equally among the users a UAV serves, $f_{m,u}(t)=f^{\max}_m/|\mathcal{U}_m(t)|$, where $f^{\max}_m$ is UAV $m$'s total CPU frequency and $\mathcal{U}_m(t)=\{u\in\mathcal{K}(t):x_{u,m}(t)=1\}$ is the set of users it serves in slot $t$---a deliberate simplification; $f_{m,u}$ can be promoted to a fourth decision without changing the framework.

\subsection{Air--ground channel}
The channel power gain $h_{u,m}(t)$ follows the standard probabilistic air--ground model~\cite{alhourani2014lap}: the elevation-dependent LoS probability of~\cite{alhourani2014lap} mixes LoS and NLoS states, which share a free-space path loss and differ by a state-dependent excess attenuation $\eta_X$, by the variance of their log-normal shadowing, and by their Nakagami fading parameter; the channel is quasi-static within a slot. Since this component is standard, we omit its equations and take its constants from~\cite{alhourani2014lap}, whose urban parameter set we adopt; the link operates at a 2\,GHz carrier, and the rate below uses $B{=}1$\,MHz, $P_u{=}0.1$\,W and $\sigma^2{=}-110$\,dBm. Each served user transmits on its own channel of bandwidth $B$, so the uplink rate is
\begin{equation}
R_{u,m}(t)=B\,
\log_2\!\Big(1+\frac{P_u\,h_{u,m}(t)}{\sigma^2}\Big),
\label{eq:rate}
\end{equation}
with user transmit power $P_u$ and noise power $\sigma^2$. Bandwidth contention among the users of one UAV is a deliberate simplification: the cost of adding a user to a UAV falls on the shared edge compute alone.

\subsection{Task and computation model}
Task $u$ is described by the tuple $(D_u,C_u,\tau^{\mathrm{dl}}_u)$: input size $D_u$ bits, computing density $C_u$ cycles/bit, and deadline $\tau^{\mathrm{dl}}_u$; our instantiation holds $C_u$ at a single value (Appendix~A). The split $\rho_u$ executes the two parts in parallel, the local part on the user's own CPU at frequency $f^{\mathrm{loc}}_u$,
\begin{align}
T^{\mathrm{loc}}_u(t)&=\frac{(1-\rho_u(t))\,D_uC_u}{f^{\mathrm{loc}}_u},
\label{eq:tloc}\\
T^{\mathrm{off}}_u(t)&=\frac{\rho_u(t)\,D_u}{R_{u,m}(t)}+\frac{\rho_u(t)\,D_uC_u}{f_{m,u}(t)},
\label{eq:toff}
\end{align}
with task completion $T_u(t)=\max\{T^{\mathrm{loc}}_u(t),T^{\mathrm{off}}_u(t)\}$, where \eqref{eq:toff} sums uplink transmission and edge execution (result feedback is negligible, as standard). Tasks settle within their arrival slot (no cross-slot queueing in this model); a task missing its deadline counts as a violation. The per-slot delay increment $\Delta T(t)$ sums completion times over active users plus a penalty $\chi$ per deadline violation, with $\chi$ sized so that penalties do not dominate the controllable part of the objective.

\subsection{Energy model}
\label{sec:energy}
UAV propulsion uses the full rotary-wing model of~\cite{zeng2019uavproc},
\begin{multline}
P^{\mathrm{fly}}(v)=\underbrace{P_0\Big(1+\frac{3v^2}{U_{\mathrm{tip}}^2}\Big)}_{\text{blade profile}}
+\underbrace{P_i\Big(\sqrt{1+\frac{v^4}{4v_0^4}}-\frac{v^2}{2v_0^2}\Big)^{\!1/2}}_{\text{induced}}\\
+\underbrace{\tfrac12\,d_0\,\rho_a\,s_r A\,v^{3}}_{\text{parasite}},
\label{eq:prop}
\end{multline}
with the standard rotor and airframe constants of~\cite{zeng2019uavproc} ($P_0{=}79.9$\,W blade profile, $P_i{=}88.6$\,W induced, $U_{\mathrm{tip}}{=}120$\,m/s, $v_0{=}4.03$\,m/s, and parasite coefficient $\tfrac12 d_0\rho_a s_rA{=}0.02$). Two properties matter. The parasite $v^3$ term: truncating \eqref{eq:prop} after two terms makes hovering and cruising nearly isoenergetic, destroying the energy--delay tension the problem is about. And the full curve's \emph{interior minimum} at $v^*{=}8.4$\,m/s (hovering costs 25\% more; 30\,m/s costs $4.8\times$ more): the physically optimal way to save energy is to fly \emph{at the valley} and trade through offloading---and which methods find, hold, or abandon this point explains most of the performance ordering in our study (Sec.~\ref{sec:corpusadv}). Edge computation obeys the standard cubic power law, at $E^{\mathrm{comp}}_{m,u}(t)=\kappa_c f_{m,u}(t)^{2}\rho_u(t)D_uC_u$ per offloaded task, with $\kappa_c$ the effective switched capacitance of the edge processor, and the per-slot system energy increment $\Delta E(t)$ sums flight power over the slot and edge-computation energy over each UAV's served users (UAV platform energy; user-side consumption is not an objective here).

\subsection{Constrained multi-objective scheduling problem}
\label{sec:program}
Collecting the models above, the offline design problem is
\begin{align}
(\mathrm{P1})\ \ \min_{\{\mathbf{a}_t\}_{t=1}^{T}}\ \ &\Big(T_{\mathrm{tot}},\,E_{\mathrm{tot}}\Big)
=\Big(\sum_{t=1}^{T}\Delta T(t),\ \sum_{t=1}^{T}\Delta E(t)\Big)
\label{eq:p1}\\
\text{s.t.}\ \
&\eqref{eq:kin},\quad 0\le d_m(t)\le D_{\max},\tag{C1}\\
&\lVert\mathbf{q}_m(t)-\mathbf{q}_{m'}(t)\rVert\ge D_{\min},\ \ \forall m\neq m',\tag{C2}\\
&\textstyle\sum_{m\in\mathcal{M}}x_{u,m}(t)\le 1,\ \ x_{u,m}(t)\in\{0,1\},\tag{C3}\\
&0\le\rho_u(t)\le\textstyle\sum_{m}x_{u,m}(t),\tag{C4}\\
&\textstyle\sum_{u}x_{u,m}(t)\,f_{m,u}(t)\le f^{\max}_m,\tag{C5}\\
&\textstyle\sum_{t=1}^{T}E_m(t)\le E^{\max}_m,\tag{C6}
\end{align}
where C1--C2 bound motion and enforce collision avoidance, C3--C4 couple offloading to association, C5 caps edge compute, and C6 is a \emph{long-term energy budget}, $E_m(t)$ being UAV $m$'s energy spent in slot $t$ and $E^{\max}_m$ its onboard capacity. Deadlines enter softly, through the penalty $\chi$. P1 is a non-convex mixed-integer nonlinear program whose single-objective scalarizations are already NP-hard~\cite{jtoratc}: continuous flight and offloading variables couple with binary association through the non-convex rate \eqref{eq:rate} and the non-monotone propulsion curve \eqref{eq:prop}.

C6 is where capability (B)'s budget contract lives in the mathematics, and it is the only constraint of P1 that spans the horizon: C1--C5 bind slot by slot, while C6 couples \emph{all} slots through the running sum $\sum_{t'\le t}E_m(t')$. No memoryless per-slot rule can enforce it, so honoring a budget requires the policy to carry its own cumulative consumption. The standard remedy has a cost of its own: a Lagrangian treatment bakes one budget into training and must retrain for the next.

Because the two objectives conflict with each other, P1 has no single minimizer but a \emph{Pareto set}. From here on we negate the objectives and work with the \emph{return} $\mathbf{G}=-(T_{\mathrm{tot}},E_{\mathrm{tot}})$, so that larger is better throughout the paper. A return $\mathbf{G}$ \emph{dominates} $\mathbf{G}'$ ($\mathbf{G}\succ\mathbf{G}'$) iff $G_j\ge G'_j$ for all $j$ with at least one strict inequality; the Pareto front is
\begin{equation}
\mathcal{P}=\big\{\mathbf{G}\in\mathcal{G}:\nexists\,\mathbf{G}'\in\mathcal{G},\ \mathbf{G}'\succ\mathbf{G}\big\},
\label{eq:front}
\end{equation}
with $\mathcal{G}\subset\mathbb{R}^{n_o}$ the attainable return set over $n_o$ objectives. Our goal is not one point of $\mathcal{P}$ but a \emph{single policy} able to reach any of them on request.

\subsection{Multi-objective MDP reformulation}
P1 induces a multi-objective MDP $\langle\mathcal{S},\mathcal{A},P,\mathbf{R},\Omega,f,\gamma\rangle$~\cite{roijers2013survey} with $\gamma{=}1$ over the finite horizon. The state $\mathbf{s}_t$ gathers per-node indicators---each UAV's position, residual energy, load and speed, and each active user's UAV-centric position, task tuple, queue and channel indicators; the action is \eqref{eq:action}; the vector reward is the negative objective increment, $\mathbf{r}_t=-(\Delta T(t),\Delta E(t))\in\mathbb{R}^{n_o}$, so the expected vector return $\mathbf{G}^\pi=\mathbb{E}_\pi\big[\sum_{t=1}^{T}\mathbf{r}_t\big]$ equals $-(T_{\mathrm{tot}},E_{\mathrm{tot}})$ in expectation. Preferences live on the simplex $\Omega=\{\boldsymbol{\omega}\in\mathbb{R}^{n_o}_{\ge0}:\sum_j\omega_j=1\}$ with linear utility $f(\mathbf{G},\boldsymbol{\omega})=\boldsymbol{\omega}^\top\mathbf{G}$. Not every preference is physically expressible: the achievable front has ends, and preferences past them saturate rather than producing more extreme outcomes---which band is attainable is a property of the corpus, settled in Sec.~\ref{sec:distill}. Rather than solving one scalarized instance of P1 per $\boldsymbol{\omega}$, we seek a \emph{single conditional policy}
\begin{equation}
\pi_\theta(\mathbf{a}_t\mid \mathbf{s}_{1:t},\,\boldsymbol{\omega}),\qquad \boldsymbol{\omega}\in\Omega,
\label{eq:goal}
\end{equation}
whose induced return approaches the Pareto-optimal return at every attainable preference. The aliasing regime of Prop.~\ref{prop:alias} begins at three objectives and is outside this paper's scope.

\subsection{Why not scalarize first}
\begin{proposition}[Convex-hull limitation]
\label{prop:hull}
Let $\mathcal{G}\subset\mathbb{R}^{n_o}$ be the attainable return set. For any $\boldsymbol{\omega}$ in the simplex, maximizers of $\boldsymbol{\omega}^\top\mathbf{G}$ over $\mathcal{G}$ lie on the boundary of $\mathrm{conv}(\mathcal{G})$; hence unsupported Pareto-optimal points---those not on the convex hull---are unattainable by a priori linear scalarization for every $\boldsymbol{\omega}$.
\end{proposition}
\begin{proof}
Linear objectives attain their maxima over $\mathcal{G}$ and over $\mathrm{conv}(\mathcal{G})$ at the same value, and over a compact convex set at supporting points of the boundary. A point strictly inside the hull's Pareto face has, for every $\boldsymbol{\omega}$, a strictly larger-utility supporting point, so it is never a maximizer. See also~\cite{das1997drawbacks}.
\end{proof}
Discrete offloading choices in \eqref{eq:action} and the non-monotone propulsion power curve \eqref{eq:prop} make $\mathcal{G}$ non-convex in this problem, so this is not a hypothetical concern. PrefDT sidesteps it because it never scalarizes at training time: it imitates preference-labeled behavior directly, so its reachable set is bounded by corpus coverage, not by hull geometry. The same design keeps training offline: the entire pipeline consumes preference-annotated trajectories, whereas the offline search that produces the rule those trajectories come from (the \emph{teacher} of Sec.~\ref{sec:distill}) costs 606{,}000 slots of live interaction, and a scalarized learner pays a live (re)training per preference. Sec.~\ref{sec:switchcost} does the cost accounting and states the boundary: what offline training removes is live interaction at deployment, not the value of a competent teacher having existed.

% ==================== IV. PREFDT ====================
\section{The Proposed PrefDT Scheduler}
\label{sec:prefdt}
One decision organizes this section: the setting is an input the policy reads at run time, not a term compiled into what the policy was trained to maximize. The pipeline that results runs at two time scales, and Fig.~\ref{fig:arch} separates them. One scale runs once, before the fleet flies, and ends in a frozen artifact. The other runs in every slot and updates no weights.

\begin{figure*}[t]
\centering
\includegraphics[width=\textwidth]{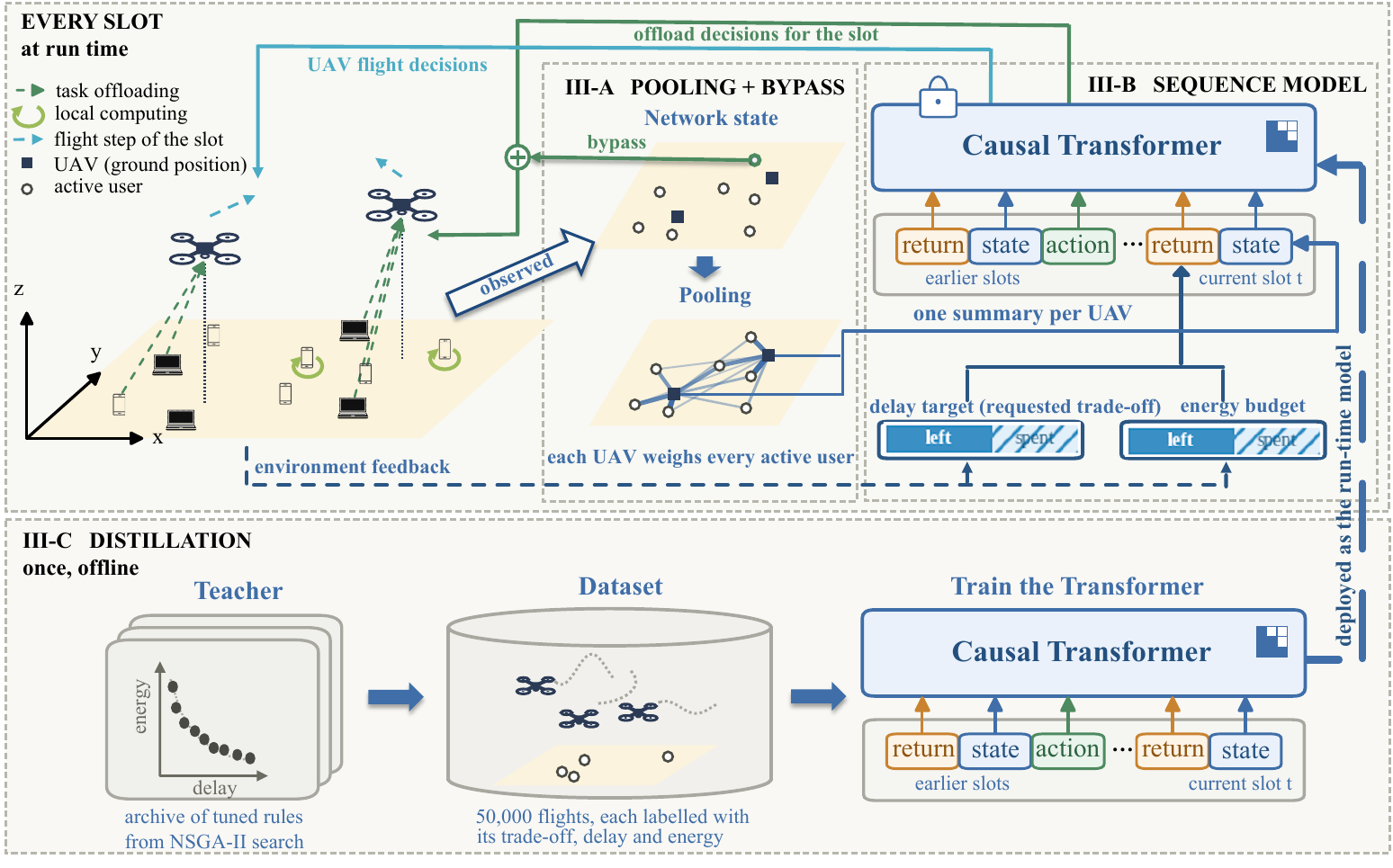}
\caption{The structure of PrefDT. The pipeline runs at two time scales. At run time and in every slot (top): the network state is read as one token per UAV and one per active user; each UAV pools every active user into one summary, and each user's own features bypass the pooling to rejoin the Transformer output ($\oplus$) for that user's decisions. A frozen causal Transformer reads the stream of return, state and action tokens and emits the slot's decisions, which the network executes; the return token's two entries, the delay target set by the requested trade-off and the energy budget, are decremented by the delay and energy the environment reports back. No weight changes at run time. Offline and once (bottom): the teacher, an archive of scheduling rules tuned by NSGA-II, is flown to record 50,000 flights, each labelled with the trade-off it struck and the delay and energy it cost; the Transformer is trained to imitate these flights and is then frozen into the model that runs above.}
\label{fig:arch}
\end{figure*}

The offline scale builds the corpus and trains on it. A teacher's rule family is searched once, and each of its rollouts is paired with the setting whose outcome that rollout realized; the pairs are the training corpus. A causal sequence model is trained on the corpus by imitation, and a regressor is fitted beside it that turns a setting into the return the corpus attains there. Training then stops. The weights, the regressor and the per-objective scales are archived together, and none of them changes again.

The per-slot scale is a loop. At the start of a flight the requested setting becomes a return target, and a requested energy budget replaces that target's energy entry. In each slot the network state becomes one token per UAV and one per active user, and each UAV pools the active users into a summary of fixed width. The summary, the preference and the current return target enter the token stream. The sequence model reads the most recent slots of that stream and returns one output per slot, from which the flight controls are decoded; each user's association and offload fraction are decoded from that output together with that user's own token. The environment executes the joint action and returns the delay and energy actually consumed in that slot. Both values are subtracted from the return target before the next slot starts. One forward pass per slot, with nothing solved or searched, and no weight updated.

\subsection{The state encoder: tokens, pooling, and the per-user path}
\label{sec:encoder}
\label{sec:token}
\label{sec:pooling}
\label{sec:peruser}
The active set $\mathcal{K}(t)$ changes size every slot, while a sequence model consumes a fixed layout of inputs. The association and offloading of \eqref{eq:action} are decided per member, so a description of the fleet as a whole is not sufficient. The encoder must accept a set of any size and still return a decision for every member.

The encoder therefore has two paths. One condenses $\mathcal{K}(t)$ into a summary of fixed width. That summary is symmetric in its members and therefore retains no identity, so no per-member decision can be read from it. The other path supplies each member's features to the heads that decide for it, without passing through the summary.

One frozen model then serves an active set of any size, and still decides per member. Changing the number of users requires no retraining, because no part of the network is sized by that number. The summary is defined on any subset of $\mathcal{K}(t)$, so a report lost on the uplink leaves a smaller set rather than an incomplete input. The scheduler therefore keeps deciding on a link that drops reports, with no substitute value invented for the members that are missing.

Each slot is encoded as one token per node. A UAV token holds that aircraft's normalized position, residual energy, number of served users and speed. A user token holds that user's position in the serving UAV's polar frame, its task tuple, its queue indicator, its channel power gain, and a block $\tilde{\boldsymbol{\psi}}_u$ of four closed-form features,
\begin{align}
\mathbf{z}^{\mathrm{uav}}_m(t)&=\big[\tilde{\mathbf{q}}_m(t),\ \tilde e^{\mathrm{res}}_m(t),\ |\mathcal{U}_m(t)|,\ \tilde v_m(t)\big],
\label{eq:tokuav}\\
\mathbf{z}^{\mathrm{usr}}_u(t)&=\big[\tilde r_{u,m},\ \tilde\phi_{u,m},\ \tilde D_u,\ \tilde C_u,\ \tilde\tau^{\mathrm{dl}}_u,\ \tilde q_u,\ \tilde h_{u,m},\ \tilde{\boldsymbol{\psi}}_u\big],
\label{eq:tokusr}
\end{align}
where every entry is standardized by statistics archived with the training set, $\tilde x=(x-\mu_{\mathcal{D}})/\sigma_{\mathcal{D}}$, so that training and rollout use identical scales. The UAV tokens number $M$ in every slot, while the user tokens number $|\mathcal{K}(t)|$ and therefore change in count from slot to slot.

The four features in $\boldsymbol{\psi}_u$ are the two latencies that decide user $u$'s current task and their two decision-relevant ratios,
\begin{equation}
\boldsymbol{\psi}_u=\big[\,T^{\mathrm{loc}}_u,\ \ T^{\mathrm{off}}_u,\ \ T^{\mathrm{loc}}_u/\tau^{\mathrm{dl}}_u,\ \ T^{\mathrm{off}}_u/T^{\mathrm{loc}}_u\,\big],
\label{eq:physfeat}
\end{equation}
with $T^{\mathrm{loc}}_u$ and $T^{\mathrm{off}}_u$ the local and offloading latencies of \eqref{eq:tloc}--\eqref{eq:toff}, evaluated at the current channel state and edge load. Both are closed-form functions of quantities the state already contains: whatever the physical models compute in closed form, the policy is not required to learn.

Each UAV condenses the active users into a single summary by masked cross-attention. For each of $H_a$ heads of width $d_h$, the query $\mathbf{q}^{(i)}_m$ is a linear projection of the UAV token, and the key $\mathbf{k}^{(i)}_u$ and the value $\mathbf{v}^{(i)}_u$ are linear projections of the user token; the attention weights are
\begin{equation}
\alpha^{(i)}_{u,m}(t)=\frac{\exp\!\big(\langle\mathbf{q}^{(i)}_m,\mathbf{k}^{(i)}_u\rangle/\sqrt{d_h}+\mu_u(t)\big)}
{\sum_{u'}\exp\!\big(\langle\mathbf{q}^{(i)}_m,\mathbf{k}^{(i)}_{u'}\rangle/\sqrt{d_h}+\mu_{u'}(t)\big)},
\label{eq:attn}
\end{equation}
with the mask $\mu_u(t)=0$ for $u\in\mathcal{K}(t)$ and $-\infty$ otherwise, so that inactive users receive exactly zero weight. Summing the values under these weights and projecting the concatenated heads gives the per-UAV summary $\mathbf{p}_m(t)$, and the encoder's representation of the slot is
\begin{equation}
s_t=\big[\mathbf{z}^{\mathrm{uav}}_1;\mathbf{p}_1(t);\ \cdots;\ \mathbf{z}^{\mathrm{uav}}_M;\mathbf{p}_M(t)\big].
\label{eq:stok}
\end{equation}
This representation has one width for every value of $|\mathcal{K}(t)|$, and it contains no padded positions. Each summary is a weighted sum over the user tokens, so it is unchanged by any permutation $\pi$ of the active users,
\begin{equation}
\mathbf{p}_m\big(\{\mathbf{z}^{\mathrm{usr}}_{\pi(u)}\}\big)=\mathbf{p}_m\big(\{\mathbf{z}^{\mathrm{usr}}_{u}\}\big),
\label{eq:perminv}
\end{equation}
and carries no information about which user occupies which index.

Each user's own token is therefore supplied directly to that user's two decision heads, concatenated with the sequence model's output $\mathbf{e}_t$ at slot $t$. Writing $x_u(t)\in\{0,1,\dots,M\}$ for user $u$'s categorical association, with $0$ denoting local execution,
\begin{align}
\boldsymbol{\lambda}_u(t)&=\mathrm{MLP}_{c}\big([\mathbf{e}_t;\ \mathbf{z}^{\mathrm{usr}}_u(t)]\big)\in\mathbb{R}^{M+1},
\label{eq:bypass1}\\
\hat\rho_u(t)&=\tfrac12\Big(1+\tanh\,\mathrm{MLP}_{\rho}\big([\mathbf{e}_t;\ \mathbf{z}^{\mathrm{usr}}_u(t)]\big)\Big),
\label{eq:bypass2}
\end{align}
while the flight controls are decoded from $\mathbf{e}_t$ alone, $\hat{\mathbf{a}}^{\mathrm{fly}}_t=\tanh\,\mathrm{MLP}_{\mathrm{f}}(\mathbf{e}_t)$. Permuting the active users permutes the outputs of \eqref{eq:bypass1}--\eqref{eq:bypass2} in the same order, so this second path is \emph{equivariant}.
\subsection{The return channel: the vector return-to-go, its decrement, and the sequence model}
\label{sec:channel}
\label{sec:inject}
\label{sec:theory}
\label{sec:backbone}
This subsection builds the value that carries the setting into the policy. Two kinds of request arrive through it. One is the energy--delay trade-off. The other is a total energy budget in joules. Both are chosen before the flight and cannot be derived by the policy from the network state, so both have to be supplied as an input the policy reads. The second one also has to change as the fleet spends.

The value is a vector with one entry per objective rather than a single total. Each entry is the return still to be collected on that objective, so the energy entry can be initialized at the requested budget and reduced at every slot by what that slot consumed. The sequence model reads a window of slots rather than a single slot, so it sees that reduction as it happens. A single total does neither of these things. Under three or more objectives it names a whole family of front points at once, so it cannot say which of them is being asked for. And at any number of objectives, including two, no part of it is a quantity in joules, so a budget can neither be requested through it nor reduced inside it.

Every setting therefore names a single outcome the fleet can produce. A total energy budget can be requested in joules and enforced against what the fleet has actually spent: an overspend enters the value the policy reads one slot after it begins, with nothing re-issued and no outer loop.

The conditioning value at slot $t$ is the per-objective return-to-go, the suffix sum $\mathbf{g}_t=\sum_{t'\ge t}\mathbf{r}_{t'}\in\mathbb{R}^{n_o}$, the componentwise analogue of DT's scalar return-to-go~\cite{chen2021dt,zhu2023peda}. Raw delay and energy returns differ by about two orders of magnitude, so each entry is rescaled by a constant of its own,
\begin{equation}
\hat g_{t,j}=\frac{g_{t,j}}{\eta_j},\qquad
\eta_j=\max_{\mathcal{D}}g_{1,j}-\min_{\mathcal{D}}g_{1,j},
\label{eq:rtgnorm}
\end{equation}
where $\boldsymbol{\eta}$ collects the constants and $\oslash$ denotes element-wise division. The rescaling is a pure ratio: no shift is subtracted, and no normalization layer follows the projection. The constants are per-objective, fixed, and archived with the dataset, so that training and rollout use identical scales.

The preference enters the token stream at three sites,
\begin{equation}
s^*_t=s_t\oplus\boldsymbol{\omega},\qquad
a^*_t=a_t\oplus\boldsymbol{\omega},\qquad
g^*_t=\hat{\mathbf{g}}_t\odot\boldsymbol{\omega},
\label{eq:inject}
\end{equation}
by concatenation to the state and action tokens and by element-wise weighting of the normalized return-to-go~\cite{zhu2023peda}. Concatenation before any layer makes the preference part of every position's representation, rather than a separate token that attention may or may not route to. The mechanism is PEDA's, unchanged; what PrefDT adds on top of it is the pipeline of Sec.~\ref{sec:pipeline}, and Sec.~\ref{sec:attribution} removes each addition singly.

Trajectories are serialized in DT's returns-first token order $(g^*_t,s^*_t,a^*_t)$~\cite{chen2021dt}, so a context of $K_{\mathrm{ctx}}$ slots spans $3K_{\mathrm{ctx}}$ tokens, each summed with a learned per-slot timestep embedding. A causal GPT backbone (3 layers, 4 heads, width 256, $K_{\mathrm{ctx}}{=}20$) predicts $\mathbf{a}_t$ at the position of $s^*_t$, and its output at that position is the $\mathbf{e}_t$ of \eqref{eq:bypass1}--\eqref{eq:bypass2}.

Training minimizes DT's standard mixed imitation loss~\cite{chen2021dt}: L2 on the continuous actions, the flight controls and the offload ratios, and cross-entropy on the per-user association, equally weighted. Actions are mapped to $[-1,1]$ beforehand by the standard rescaling $\tilde{\mathbf{a}}=2(\mathbf{a}-\mathbf{a}_{\min})/(\mathbf{a}_{\max}-\mathbf{a}_{\min})-1$. AdamW at $10^{-4}$ with warmup, batch 256; convergence in 40k steps. The policy head is deterministic by default, and sampling the association head is an inference-time option rather than a training change: it draws from the existing softmax, and no parameter is refit.

Across the slots of one context the return-to-go decrements by the reward each slot realized,
\begin{equation}
\hat{\mathbf{g}}_{t+1}=\hat{\mathbf{g}}_t-\mathbf{r}_t\oslash\boldsymbol{\eta},
\label{eq:decr}
\end{equation}
so an energy entry requested at a budget, $\hat g_{1,E}=-\mathcal{B}/\eta_E$, stands after $t$ slots at
\begin{equation}
\hat g_{t,E}=-\frac{\mathcal{B}-\sum_{t'<t}\Delta E(t')}{\eta_E}=-\frac{\mathcal{B}-\mathrm{cum}E(t)}{\eta_E}.
\label{eq:identity}
\end{equation}
The energy entry of the conditioning vector is the budget not yet spent, up to the fixed scale $\eta_E$. A revision of the total to $\mathcal{B}'$ is a single assignment to \eqref{eq:identity}, and requires no further change to the loop.

How many entries the conditioning value needs is set by the dimension of the front.
\begin{proposition}[Utility aliasing]
\label{prop:alias}
Let $\mathcal{P}\subset\mathbb{R}^{n_o}$ be a compact $(n_o{-}1)$-manifold and let $\boldsymbol{\omega}$ have strictly positive entries. The level set $\{\mathbf{g}:\boldsymbol{\omega}^\top\mathbf{g}=c\}$ meets $\mathcal{P}$ transversally in a set of dimension $n_o{-}2$: a point for $n_o{=}2$, and a positive-dimensional family for $n_o\ge3$.
\end{proposition}
The proof is a transversality argument on the dimension of the intersection; the sketch is in Appendix~B. One entry therefore names a single front point at two objectives and a whole family of them at three or more, while $n_o$ entries name one point at every $n_o$.

Two further differences hold at every $n_o$, including the $n_o{=}2$ case the proposition leaves sufficient. A scalar utility has no component that is the energy budget, so the identity \eqref{eq:identity} is available to the vector form alone. And the vector form supplies $n_o$ supervision targets per timestep where the scalar supplies one, with the decrement \eqref{eq:decr} propagating per-objective rather than aggregate progress.
\subsection{The distillation pipeline: corpus construction and inference}
\label{sec:pipeline}
\label{sec:distill}
\label{sec:infer}
Conditional imitation needs recorded runs labeled with the outcome each achieved, and scheduling supplies none: no archive of flights exists in which each flight carries the trade-off it struck. The corpus therefore has to be produced. Which flight is stored under which setting is what a setting comes to mean, and no loss function decides that afterwards.

The corpus is distilled from a teacher: a compact physical rule whose ten parameters are searched once, offline, and whose rollouts become the training data. A policy trained by imitation can only produce behaviors its corpus contains, so the teacher sets the ceiling, and choosing it is part of the algorithm rather than a tuning detail. That choice is made on physics: the cheapest way to save energy is to hold the speed at the propulsion power minimum and trade through partial offloading, so a teacher qualifies only if its members sweep the front that way. Each setting is then paired with the member whose realized cost composition matches it, which is what makes a setting an outcome share rather than an importance weight over a training reward. Settings outside the band that family spans are refused before any data is generated, and the extreme members are demonstrated more heavily than uniform sampling would demonstrate them, because imitation regresses toward the behavioral mean.

The search is paid once. Afterwards a frozen model answers every setting the archive could have answered, and three requests the archive could not: a point between two of its members, a total energy budget, and a revision of that budget in flight. The setting axis is calibrated onto the front the trained model actually achieves, so distinct settings name distinct attainable outcomes in an exact order and at the requested density. A denser front costs further rollouts and no further search.

The shipped corpus is distilled from the teacher's archive: rollouts of the NSGA-II-searched rule family of Sec.~\ref{sec:classes}, whose members hold $v^*$ throughout and sweep the front through the offload ratio. The archive also fixes which preferences are physically expressible. With $\hat{\mathbf{G}}(\boldsymbol{\omega})$ the return the family realizes at $\boldsymbol{\omega}$, the attainable band is
\begin{equation}
\Omega^*=\Big\{\boldsymbol{\omega}\in\Omega:\
\varsigma_1\big(\hat{\mathbf{G}}(\boldsymbol{\omega})\big)\in[0.052,\ 0.676]\Big\},
\label{eq:omegastar}
\end{equation}
where $\varsigma_1(\cdot)$ is the delay share of the normalized return and the interval is the span the family covers. Preferences outside $\Omega^*$ are rejected at data generation, so the corpus never annotates a setting the physics cannot express.

Dataset generation follows the D4MORL recipe~\cite{zhu2023peda} in three steps: sample preferences from Dirichlet distributions at three entropy tiers ($\alpha\in\{1,3,8\}$, one third of the draws each) and reject $\boldsymbol{\omega}\notin\Omega^*$, the band of \eqref{eq:omegastar} widened by a tolerance of $0.05$; match each accepted $\boldsymbol{\omega}$ to the archive member whose realized return \emph{share} is nearest,
\begin{equation}
i^\star(\boldsymbol{\omega})=\arg\min_{i}\ \big\lVert\boldsymbol{\varsigma}\big(\mathbf{G}^{\pi_i}\big)-\boldsymbol{\varsigma}\big(\hat{\mathbf{G}}(\boldsymbol{\omega})\big)\big\rVert_2;
\label{eq:match}
\end{equation}
and roll out $\pi_{i^\star}$, storing the preference-annotated trajectory---$50{,}000$ trajectories in all. Because the match is on realized share, $\boldsymbol{\omega}$ denotes an outcome composition: a front point whose normalized cost is $\omega_1$ parts delay. Two refinements act at the extremes. Corner emphasis redirects a quarter of the draws into the share window of the archive's four lowest-delay members. Quality diversity mixes archive-member and scripted-heuristic rollouts in an 80/20 ratio, the heuristic pool spanning three fixed modes, which varies outcome quality at fixed $\boldsymbol{\omega}$.

Three maps stand between a requested setting and the fleet's behavior: a conditioner $\boldsymbol{\omega}\mapsto\hat{\mathbf{g}}_1$, a remap $c\mapsto\boldsymbol{\omega}$, and the decrement loop that carries $\hat{\mathbf{g}}_1$ through the horizon. The conditioner is a preference-to-return regressor fitted offline on the corpus's clean expert episodes and frozen with the weights,
\begin{equation}
f_{\mathrm{reg}}=\arg\min_{\mathbf{A}}\sum_{(\boldsymbol{\omega}_i,\mathbf{G}_i)\in\mathcal{D}_{\mathrm{exp}}}\big\lVert\mathbf{G}_i-\mathbf{A}\,\phi(\boldsymbol{\omega}_i)\big\rVert_2^{2},
\label{eq:reg}
\end{equation}
with quadratic features $\phi$. Its targets are returns the corpus contains at that preference, never an ideal point, which lies outside the attainable set.

Two gates run before any training: the binned map $\boldsymbol{\omega}\mapsto\mathbb{E}[\mathbf{G}\,|\,\boldsymbol{\omega}]$ is checked for monotonicity, and \eqref{eq:reg} is checked for fit. Each conditioner is fitted per corpus and archived with it, and a corpus build that would overwrite another corpus's conditioner aborts by construction.

The remap is fitted after training. It is a monotone map $c\mapsto\boldsymbol{\omega}(c)$ carrying an abstract setting $c\in[0,1]$ onto the measured, non-dominated portion of the front the trained model achieves, and it is pure inference: no weight, conditioner or checkpoint changes.

At rollout the normalized return-to-go initializes from the conditioner,
\begin{equation}
\hat{\mathbf{g}}_1=f_{\mathrm{reg}}(\boldsymbol{\omega})\oslash\boldsymbol{\eta},
\label{eq:init}
\end{equation}
and from there decrements by \eqref{eq:decr}; actions decode deterministically. Algorithm~\ref{alg:sweep} composes the three maps into one sweep. It takes the frozen $\pi_\theta$, a setting grid $\mathcal{C}\subset[0,1]$, the remap, the conditioner and the scales $\boldsymbol{\eta}$, and returns the non-dominated subset of the returns the sweep realizes. Its outer loop turns one setting into one target: the remap and the conditioner are applied once, and a requested budget is written into $\hat g_E$ in place of the conditioner's energy target. Its inner loop carries that target through the horizon: one forward pass per slot, the decrement of \eqref{eq:decr}, and the same $\hat g_E$ open to a mid-flight revision. A $K$-point front therefore costs $KT$ forward passes against one frozen $\pi_\theta$.
\begin{algorithm}[t]
\caption{PrefDT front generation with a recalibrated setting axis (generalizing DT's return-conditioned evaluation loop~\cite{chen2021dt})}
\label{alg:sweep}
\begin{algorithmic}[1]
\Require frozen $\pi_\theta$; setting grid $\mathcal{C}\subset[0,1]$; remap $c\mapsto\boldsymbol{\omega}(c)$ (fitted on calibration seeds); conditioner $f_{\mathrm{reg}}$; scales $\boldsymbol{\eta}$
\For{$c\in\mathcal{C}$}
  \State $\boldsymbol{\omega}\gets\boldsymbol{\omega}(c)$;\ \ $\hat{\mathbf{g}}\gets f_{\mathrm{reg}}(\boldsymbol{\omega})\oslash\boldsymbol{\eta}$\ \ (budget override: $\hat g_E\gets -\mathcal{B}/\eta_E$);\ \ reset env;\ \ $\xi\gets\emptyset$
  \For{$t=1,\dots,T$}
    \State build $(g^*_t,s^*_t)$ from \eqref{eq:inject} with $\hat{\mathbf{g}}_t=\hat{\mathbf{g}}$; append to $\xi$
    \State $\mathbf{a}_t\gets\pi_\theta\big(\cdot\mid\xi_{t-K_{\mathrm{ctx}}+1:t},\boldsymbol{\omega}\big)$; step env $\rightarrow\mathbf{r}_t$
    \State $\hat{\mathbf{g}}\gets\hat{\mathbf{g}}-\mathbf{r}_t\oslash\boldsymbol{\eta}$\ \ \emph{(revision = overwrite $\hat g_E$; \eqref{eq:identity})}
  \EndFor
  \State record $\mathbf{G}(c)=\sum_{t}\mathbf{r}_t$
\EndFor
\State \Return non-dominated subset \eqref{eq:front} of $\{\mathbf{G}(c)\}_{c\in\mathcal{C}}$
\end{algorithmic}
\end{algorithm}

\section{Experimental Setup}
\label{sec:setup}

\subsection{Scenario}
\label{sec:scenario2}
The environment implements Sec.~\ref{sec:model}: $M{=}2$ UAVs at $H{=}100$\,m over a $1000{\times}1000$\,m area, $U{=}10$ users of whom four to eight are active in any slot, and $T{=}100$ slots of 1\,s. Tasks, deadlines, compute rates, partial offloading, rotary-wing propulsion and the probabilistic LoS air--ground channel are as modelled there, with every constant listed in Appendix~A. All methods share this environment and identical evaluation seeds, with one exception: the degradation study of Sec.~\ref{sec:budget} adds a mid-episode propulsion multiplier, pinned so that it cannot have perturbed anything else (Sec.~\ref{sec:configs}).

\subsection{Methods under comparison: three declared classes}
\label{sec:methods}
\label{sec:classes}
Twenty-six method variants are evaluated under one shared protocol. The comparison set spans methods with incompatible access assumptions and roles, so we declare three classes before any of their numbers are quoted (Table~\ref{tab:classes}). Table~\ref{tab:family} reports the variants the main text reads; Appendix~C lists every variant with its configuration.

\begin{table}[t]
\centering
\caption{The three declared classes of the comparison set.}
\label{tab:classes}
\scriptsize
\setlength{\tabcolsep}{3pt}
\begin{tabular}{@{}p{0.9cm}p{3.1cm}p{1.9cm}p{1.7cm}@{}}
\toprule
Class & Members & Provenance & A new trade-off costs\\
\midrule
1 Learned schedulers & PrefDT; $\boldsymbol{\omega}$-conditioned IQL~\cite{kostrikov2022iql}; PGMORL~\cite{xu2020pgmorl}; PEDA~\cite{zhu2023peda} ported to this problem; per-preference PPO; $\boldsymbol{\omega}$-BC; unconditioned DT; Lagrangian CMDP~\cite{tessler2019rcpo}; single-toggle ablations of Sec.~\ref{sec:prefdt} & Three published methods run by us; the rest built by us & One rollout if conditioned; one retraining per preference or per budget otherwise\\
\addlinespace[2pt]
2 Model-based reference bounds & SCA oracle, supplied the full future arrival sequence; SCA-MPC over a five-slot window & Literature method~\cite{wu2018multiuav,jtoratc}, our instantiation & One fresh solve per setting\\
\addlinespace[2pt]
3 Distillation teacher & Teacher: NSGA-II~\cite{deb2002nsga2} over our ten-parameter rule; the source of the training corpus & Search operator published, rule ours & One re-search; front size is the population\\
\bottomrule
\end{tabular}
\end{table}

\emph{Class 1} is the class the front-quality claim is scoped to. $\boldsymbol{\omega}$-conditioned IQL is trained on \emph{our own} corpus, so that the comparison isolates architecture from data, and the Lagrangian CMDP is retrained once per budget level.

\emph{Class 2} is instantiated over our own channel and propulsion models, at the corpus's per-objective scales. The SCA oracle is restarted from hover, from the demand centroid, and from jittered variants of each, since BCD is local and a bound is only a bound if it was given a fair search; SCA-MPC's iteration budget is deliberately smaller than the oracle's, because an MPC given the oracle's budget would not be the method this baseline exists to represent. Neither can take a setting: a new preference is a fresh solve rather than a fresh token.

\emph{Class 3} exists because its rule is \emph{our own} composition. No published UAV-MEC baseline exposes a search space small enough for a 2000-evaluation budget, so we composed one from standard elements of this literature---a latency-comparison offload gate, a service-radius association rule, and centroid-tracking motion (four global genes plus three per-UAV motion gains, ten in all)---and searched it with NSGA-II, whose published contribution is the search operator rather than the rule.
\subsection{Metrics and protocol}
\label{sec:protocol}
\label{sec:dissoc}
Front quality is hypervolume (HV~\cite{zitzler1999hv}) with IGD, both per-seed against a reference point and reference front shared by every method compared.

Front quality does not measure whether settings are followed---a scheduler can produce an excellent front while ignoring every request---so the conditioning channel carries its own pair of readings. A setting is followed when different settings produce different outcomes, and the outcomes arrive in the requested order. \emph{Fidelity} measures the order: per seed, the rank correlation $|\rho|$ between the settings and the front positions they produce, with $|\rho|{=}1$ meaning every outcome arrives exactly in the requested order. Fidelity alone is insufficient, because perfect ordering can be held over a narrow segment of the front. \emph{Reach} closes that gap: the span, in seconds of delay, of the outcomes the settings produce on one seed. Reach alone is insufficient in turn, because a wide span can be produced with no order at all. Both defects occur among the systems evaluated in Sec.~\ref{sec:results}, so the two readings are used only together, against a criterion fixed before the decisive experiments:
\begin{equation}
\text{operable}\iff |\rho|\ge0.9\ \wedge\ \text{reach}\ge0.8\cdot W_{\mathrm{teacher}},
\label{eq:criterion}
\end{equation}
where $W_{\mathrm{teacher}}$ is the span that the teacher's archive itself attains per seed. The requirement therefore reads: near-perfect ordering, over at least $80\%$ of the span the teacher covers---$493$\,s in this study. Both clauses are per-seed statistics, because mixing a pooled threshold with a per-seed statistic biases the test. Neither constant was chosen after the results were seen; the pair is this study's operating point, not a proposed universal. $|\rho|$ and reach are reported throughout on the full uniform-$\boldsymbol{\omega}$ grid; what the recalibrated setting axis of Sec.~\ref{sec:infer} adds is reported in Sec.~\ref{sec:corpus}.

A budget's violation rate is read with its tracking error, since a policy that ignores its setting downward scores a flattering violation rate. Every inferential comparison uses 100 paired evaluation seeds and a two-sided Wilcoxon signed-rank test with the exact null distribution, computed by a subset-sum dynamic program, since Monte-Carlo nulls cannot resolve the $10^{-7}$-scale $p$-values a paired design at this sample produces; relative effects carry bootstrap confidence intervals. Appendix~C gives the machinery and the controls behind it.

% ==================== VII. RESULTS ====================
\subsection{Experiment configurations}
\label{sec:configs}
How every experiment in Sec.~\ref{sec:results} was run---its grid, its perturbation and its seed count---is tabulated in Appendix~C of the supplementary material; seven conventions hold throughout and are not repeated there. An ablation toggles one component at a time, with architecture, training schedule and seeds otherwise fixed. Readings on the PPO-corpus configuration and on the archive corpus, which is the shipped path, are never differenced against each other. Every arm is evaluated on a frozen checkpoint with no retraining. Inferential rows are read at $n{=}100$; descriptive rows carry the $n$ the table prints, and the two are never differenced against each other---where a headline comparison needs a row that Table~\ref{tab:family} prints as descriptive, that row is re-evaluated at $n{=}100$ first. Calibration and evaluation seeds are disjoint wherever a fitted component---the recalibration remap---could otherwise be scored on its own fitting data. Hypervolume is read against one shared reference point and pooled reference front, recomputed whenever the method set changes; the reference file carries a hash of the contributing method set, so silent drift is structurally impossible. And the degradation hook of Sec.~\ref{sec:budget} reproduces the unmodified simulator bit for bit at multiplier 1.0---a test on exact array equality, not a tolerance---so the one study that perturbs the environment cannot have moved a number in any other. The cost study alone is not evaluated on the shared grid: it is timed on one Apple M1 core with \texttt{OMP\_NUM\_THREADS}, \texttt{MKL\_NUM\_THREADS} and \texttt{VECLIB\_NUM\_THREADS} set to 1, the routes run one after another and the machine otherwise idle, the unit being one setting---one episode of $T$ slots for PrefDT and for both solvers, one 2M-step training run for the scalarized learner, one complete search for the teacher---so that a $K$-point front is $K$ times the per-setting value and is never timed as a whole. Everything else ran on one rented cloud instance (one NVIDIA RTX 4090D, 24\,GB, on a 24-vCPU allocation of a 256-core host); the Apple M1 machine served development, smoke tests, and the timing study alone.

\section{Experimental Evaluation}
\label{sec:results}

\subsection{Front quality, fidelity, and reach: comparing the field and selecting from it}
\label{sec:jointreading}
\label{sec:incumbent}
Three quantities decide a scheduler here, and no one of them decides it alone. Front quality says how good the trade-off curve is; fidelity says whether the outcomes arrive in the order the settings requested; reach says over how much of the curve that order holds. Table~\ref{tab:family} reads every method in the study on all three. Read together, PrefDT stands best: it produces the best trade-off curve of any learned method, and on the corpus it ships with it is the only learned method that clears both clauses of \eqref{eq:criterion}.

\begin{table}[t]
\centering
\caption{Methods under one shared protocol, grouped by the classes of Sec.~\ref{sec:classes}. HV ($\times10^6$, mean$\pm$std) and IGD are read against a reference point and reference front shared by every row; lower IGD is better. $|\rho|$ and reach are co-reported per Sec.~\ref{sec:protocol}, and the last column reads each row against the pre-registered criterion~\eqref{eq:criterion} and names the clause it fails.}
\label{tab:family}
\scriptsize
\setlength{\tabcolsep}{2.2pt}
\begin{threeparttable}
\resizebox{\columnwidth}{!}{%
\begin{tabular}{@{}lcccccc@{}}
\toprule
Method & HV & IGD & $|\rho|$ & Reach & $n$ & \eqref{eq:criterion}\\
\midrule\multicolumn{7}{@{}l}{\emph{Class 1: learned, preference-conditioned}}\\
PrefDT (shipped)$^{\dagger e}$ & \textbf{99.47}$\pm$1.18 & 402 & 0.93 & 601 & 100 & \textbf{yes}\\
PrefDT (PPO corpus)$^{\dagger}$ & 90.28$\pm$2.48 & 1457 & 0.91 & 540 & 100 & yes\\
PEDA, as published$^{\dagger g}$ & 82.87$\pm$2.17 & 1594 & 0.76 & 398 & 100 & no (both)\\
PEDA + attention pooling$^{\dagger g}$ & 84.18$\pm$1.98 & 1298 & 0.68 & 491 & 100 & no (both)\\
IQL+$\boldsymbol{\omega}$ (archive corpus)$^{c}$ & 96.23$\pm$1.82 & 1587 & 0.94 & 413 & 10 & no (reach)\\
IQL+$\boldsymbol{\omega}$ (PPO corpus)$^{c}$ & 91.83$\pm$1.48 & 1221 & 0.95 & 528 & 10 & yes\\
PGMORL & 87.34$\pm$1.23 & 4501 & 0.84 & 15 & 10 & no (both)\\
Per-preference PPO$^{\dagger}$ & 90.72$\pm$1.54 & 1369 & 0.95 & 396 & 100 & no (reach)\\
$\boldsymbol{\omega}$-BC$^{\dagger}$ & 89.23$\pm$2.88 & 1237 & 0.49 & 458 & 100 & no (both)\\
PrefDT, ideal-point targets & 90.21$\pm$2.21 & 1475 & 0.98 & 170 & 10 & no (reach)\\
\midrule
\multicolumn{7}{@{}l}{\emph{Class 2: model-based reference bounds}}\\
SCA oracle$^{a}$ & 99.91$\pm$1.25 & 314 & 0.88 & 607 & 10 & no ($|\rho|$)\\
SCA-MPC$^{b}$ & 99.94$\pm$1.27 & 597 & 0.93 & 606 & 10 & yes$^{f}$\\
\midrule
\multicolumn{7}{@{}l}{\emph{Class 3: distillation teacher (ours; corpus source)}}\\
Teacher$^{d}$ & 99.80$\pm$1.31 & 351 & --- & 616 & 10 & ---\\
\midrule
\multicolumn{7}{@{}l}{\emph{Floors}}\\
Unconditioned DT & 79.96$\pm$5.26 & 3162 & --- & 100 & 10 & ---\\
Random & 62.40 & 18529 & --- & 193 & 1 & ---\\
\bottomrule
\end{tabular}}
\TPTnoteswidth{\columnwidth}
\begin{tablenotes}[flushleft]\scriptsize
\item[] Reach is defined for any method that offers more than one front point; $|\rho|$ is a correlation between requested settings and the outcomes they produce and is undefined (---) for a method that takes no preference input.
\item[$^{\dagger}$] Inferential rows, $n{=}100$ paired seeds; the others are descriptive, $n{=}10$.
\item[$^{a}$] Non-causal: supplied the full future arrival sequence.
\item[$^{b}$] Not real-time schedulable as implemented (Sec.~\ref{sec:switchcost}; qualified in Sec.~\ref{sec:conclusion}).
\item[$^{c}$] Printed descriptive, but this row's verdict and the comparison of Sec.~\ref{sec:corpusadv} are read at $n{=}100$ (Appendix~C).
\item[$^{d}$] The search operator is published~\cite{deb2002nsga2}; the rule it searches is our own composition, and this row is the source of the training corpus.
\item[$^{e}$] Scored against an archived reference pool that does not contain it (Appendix~C).
\item[$^{f}$] Meets the numeric criterion, but its setting is a fresh solve per point rather than a run-time input (note~$^{b}$).
\item[$^{g}$] PEDA~\cite{zhu2023peda} ported to this problem on the PPO corpus with its own conditioning kept: a fixed-dimension input by zero padding and no per-user path; the second row adds the pooled encoder of Sec.~\ref{sec:encoder} and still omits the per-user path.
\end{tablenotes}
\end{threeparttable}

\end{table}

Three configurations reach higher front quality, and the three quantities do not all read for them. The teacher offers the widest span in the study, $616$\,s---the reference the reach clause is set against, since $0.8\times616$ is the required $493$---and no fidelity, since fidelity is a correlation between requested settings and the outcomes they produce and the teacher's archive is not requested. The non-causal oracle reads on both and fails the ordering clause, at $|\rho|{=}0.88$. The rolling-horizon solver reads on both and passes, at $|\rho|{=}0.93$ over $606$\,s. What the three share is the narrow band they occupy: they span $0.14$ HV, and the full-information oracle buys only $0.11$ over the ten-parameter rule. Front quality in this environment saturates against physical structure that ten parameters already capture, so it has little room left in which to separate methods, and the separation falls to the other two quantities.

The closest published method reads on the same corpus as the PPO-corpus row. PEDA~\cite{zhu2023peda} ported as published---a fixed-dimension input by zero padding, association decoded from the pooled token alone, its own conditioning kept---scores $82.87$ against PrefDT's $90.28$ on the same hundred seeds ($p{=}3.2{\times}10^{-30}$), at $|\rho|{=}0.76$ over $398$\,s, failing both clauses. Adding the pooled encoder without the per-user path recovers $1.31$; the remaining $6.10$ is the per-user bypass, which Table~\ref{tab:ablation} isolates. What separates PrefDT from the method it descends from is the interface to the physical network, not the sequence model.

PrefDT retains $99.7\%$ of its teacher's front quality, and the teacher that figure is measured against is not the one it learned from but the strongest we could find. The teacher's population is also the size of the front it returns, and re-run at the same ${\sim}2000$-evaluation budget the largest population is the strongest: $100.10$ at 101, against $99.79$ for the pop-50 archive the student was actually distilled from (Table~\ref{tab:incumbent}). Against that stronger teacher, PrefDT at 101 settings reaches $99.82$, while the teacher has no fidelity at all, its front size is its population, and a different point on it requires searching again. The $0.27$ HV between them is not spread over the front: at matched cardinality the $0$--$120$\,s band alone closes $127\%$ of the gap, more than all of it, because above $380$\,s PrefDT's points are the better ones. Appendix~D decomposes that corner into what the corpus withheld, what deterministic decoding costs, and a residual limit of imitation, and what survives the first two is about a second of corner depth on a $900$-s axis. Distillation therefore delivers a settable, densifiable front at $99.7\%$ of a teacher stronger than its own corpus source, and what separates them reduces to one second at one end.

\begin{table}[t]
\centering
\caption{PrefDT against the teacher's population scan. HV ($\times10^6$) is read over $n{=}100$ paired seeds against a shared reference; every NSGA-II row is a search of about 2000 policy evaluations.}
\label{tab:incumbent}
\scriptsize
\setlength{\tabcolsep}{1.6pt}
\begin{threeparttable}
\begin{tabular}{@{}>{\raggedright\arraybackslash}p{3.6cm}cccc@{}}
\toprule
Method & HV & Corner (s) & Settable & Densify\\
\midrule
NSGA-II pop 101 (best at fixed budget) & \textbf{100.096} & --- & no & re-search\\
NSGA-II pop 101$^{\ddagger}$ (2$\times$ budget) & 100.120 & --- & no & re-search\\
NSGA-II pop 21 & 99.887 & 102.8 & no & re-search\\
NSGA-II pop 50 (corpus source) & 99.786 & 102.0 & no & re-search\\
\textbf{PrefDT (shipped), $K{=}101$} & 99.824$^{\S}$ & 103.6 & \textbf{yes} & \textbf{rollout}\\
PrefDT (shipped), $K{=}21$ & 99.466 & 104.3 & yes & rollout\\
PrefDT (archive corpus, no corner/physics), $K{=}21$ & 99.229 & 105.0 & yes & rollout\\
\bottomrule
\end{tabular}
\begin{tablenotes}[flushleft]\scriptsize
\item[] \emph{Corner} is the deepest delay a fixed setting or archive member attains, the depth a request can reach; \emph{Settable} is whether a different point on the front can be asked for at run time; \emph{Densify} is what adding points to the front costs.
\item[$^{\ddagger}$] Given twice the search budget. The corpus-source row is the teacher row of Table~\ref{tab:family}, re-evaluated here at $n{=}100$.
\item[$^{\S}$] The headline deficit is measured against the strongest configuration, pop 101: $-0.27$ HV ($p{=}1.8{\times}10^{-6}$); against pop 21 the difference is $-0.06$ ($p{=}0.12$).
\end{tablenotes}
\end{threeparttable}

\end{table}

Against the solvers the quantity cannot be read at the sample the rest of this paper uses. Their solve cost fixes it: they exist only on the first ten seeds, since an $n{=}100$ sweep is $150$--$200$ hours of single-core compute per solver (Sec.~\ref{sec:switchcost}), and on those ten seeds the comparison must be paired---PrefDT scores $99.27$ there, not its hundred-seed $99.47$. Seed for seed SCA-MPC leads by $0.67$ and the SCA oracle by $0.64$, on 7 of 10 seeds, $p{=}0.16$, $95\%$ CI $[-0.07,+1.36]$. That is an underpowered comparison, not a measured tie and not a measured loss, and we report it as such.

Read on front quality alone, the ranking selects against the property the conditioning was added for. The highest front quality in the whole ablation family ($92.44$) belongs to a variant with the return token removed and the shortest context window, at fidelity $0.643$: it covers the front without following the settings. $\boldsymbol{\omega}$-conditioned IQL, trained on the same corpus as PrefDT, holds that clause and fails the other---$|\rho|{=}0.94$ over $413$\,s, $67\%$ of the teacher's span against the required $493$, and no advantage-weighting temperature closes it---swept over two and a half orders of magnitude, the reach deficit stays at least $97$\,s at every setting. PGMORL fails both, offering six policies that collapse to about three distinct behaviors. The sharpest case is PrefDT's own: removing the preference token raises front quality above every arm of the family while fidelity and reach fail together (Sec.~\ref{sec:attribution}). Selected by front quality alone, three of these could have been the final configuration; under the criterion, none of them is. The attribution of the teacher's own quality is Appendix~D.

\subsection{Operability: setting fidelity, budget tracking, and mid-episode revision}
\label{sec:operability}
\label{sec:budget}
Operability was defined in Sec.~\ref{sec:intro} as accepting three requests at run time: a preference over delay and energy, a total energy budget, and a mid-flight revision of that budget. The first is read against the criterion of Sec.~\ref{sec:protocol}. PrefDT posts $|\rho|{=}0.93$ over a reach of $601$\,s, clearing both clauses, and it is the only learned method in the study that clears them on the corpus it ships with (Sec.~\ref{sec:jointreading}). The second and third are read against the energy coordinate of the return token, which by \eqref{eq:identity} always equals the energy still allowed to be spent.

Under unchanged physics the budget is followed and never overrun. Realized consumption is monotone in the request across the whole range, with a gain $dE/d\mathcal{B}$ between $0.75$ and $0.99$ on all four trained variants of the pipeline (the archive and corner-emphasis corpora, each with and without the physics features), and a violation rate of $0.00$ at a budget tight enough to constrain the flight. Below the floor no policy in this study can fly under, the request saturates rather than being pursued. Under a propulsion cost raised by half in mid-flight, PrefDT finishes $+0.6\%$ over its contract with no re-issue and no outer loop, where the fixed rule of matching appetite finishes $+17\%$ and violates on every episode, and a supervisory controller built to repair exactly that blindness recovers less than one point of the overshoot (Fig.~\ref{fig:consumption}). A revision is exact where it should be exact and monotone where it should be monotone: re-issuing the same total is exact to floating point---at most $4{\times}10^{-4}$\,J against the $34$--$41$\,kJ consumed---tightening tightens and relaxing relaxes, and a revision below the physical floor is ignored rather than pursued. The one other method in this study that holds a budget is a Lagrangian CMDP retrained for each budget. Read at the three budgets it was trained for, it tracks more accurately than PrefDT---$0.107$ against $0.178$---and it exists nowhere else, each of the three bought with 2M environment steps.

\begin{figure}[t]
\centering
\includegraphics[width=\linewidth]{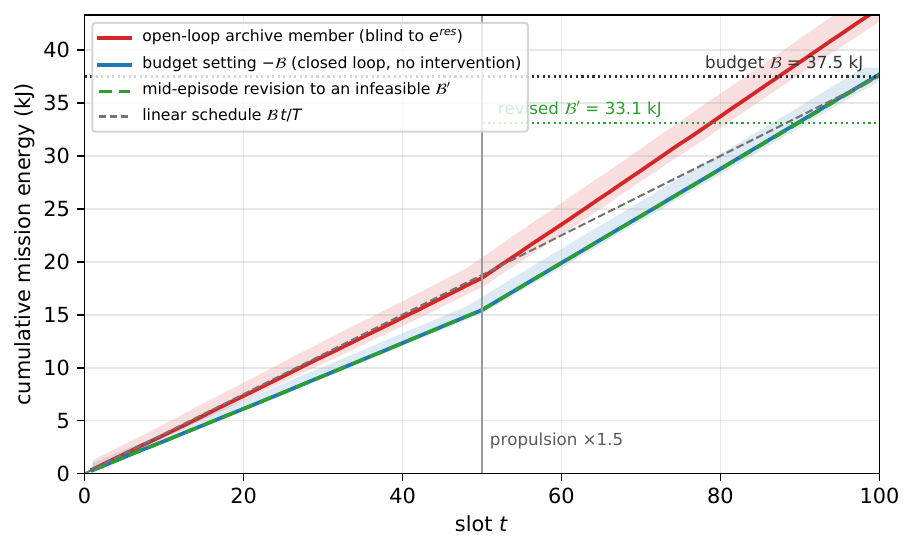}
\caption{Capabilities (B) and (C) under a mid-episode change of physics. Each trace is cumulative mission energy against the slot index, with the band one standard deviation across seeds; the horizontal lines are the requested budget and the revised total, the straight dashed line is the pace that would land exactly on the budget, and the vertical line marks the slot at which propulsion cost rises. The archive member is a fixed rule that reads no consumption state; PrefDT receives the budget as the energy coordinate of its return token and no other instruction; the third trace rewrites that coordinate mid-episode to a total below what the changed physics allows.}
\label{fig:consumption}
\end{figure}

The three readings establish that the contract is enforced against consumption rather than against a plan. The token is the remaining allowance itself, so a propulsion cost raised by half enters the quantity the policy reads one slot after the rise. The vector structure is what makes this possible: collapsed to a scalar, the same requests reverse, and asking for a larger budget yields less consumption (Sec.~\ref{sec:ladder}). The violation rate under the degraded contract is $0.70$ rather than zero, and the rate and the magnitude are consistent: the token requests a target, not a ceiling, so a policy that tracks one tightly finishes on both sides of it. A one-sided guarantee is obtained by requesting a margin instead: at $0.9\mathcal{B}$ the same frozen model runs at violation rate $0.238$, from $0.305$, at relative error $0.187$ from $0.172$.

\subsection{What a new trade-off costs: one rollout against a re-solve, a re-search, and a retraining}
\label{sec:switchcost}
A setting supplied at run time is worth what a new setting costs. Four routes to a new point on the trade-off curve are priced here against each other: one rollout of the frozen model, one convex re-solve, one fresh search of the teacher's rule, and one training run of a scalarized learner. All four are timed under the single-core protocol of Sec.~\ref{sec:configs}.

\begin{table}[t]
\centering
\caption{What one new setting costs, by route, timed under the single-core protocol of Sec.~\ref{sec:configs}.}
\label{tab:unifiedcost}
\scriptsize
\setlength{\tabcolsep}{3.5pt}
\begin{threeparttable}
\begin{tabular}{@{}lrrrc@{}}
\toprule
Route to one new setting & Interaction & Core-s & Wall-clock (s) & $n$\\
\midrule
PrefDT, one rollout & 100 slots & 0.41 & \textbf{0.41}$\pm$0.02 & 20\\
SCA oracle, one convex re-solve & 100 slots & 251.7 & \textbf{252} (155--378) & 3\\
SCA-MPC, one convex re-solve & 100 slots & 340.0 & \textbf{341}$\pm$7 & 3\\
Per-preference PPO, one training run & 2M steps & 986.7 & \textbf{988}$\pm$3 & 3\\
Teacher, one re-search$^{\ast}$ & 606{,}000 slots & 133.8 & \textbf{135} & 1\\
\bottomrule
\end{tabular}
\begin{tablenotes}[flushleft]\scriptsize
\item[] \emph{Interaction} is the simulated slots or environment steps the route consumes for one setting, and does not depend on the machine; \emph{core-s} is CPU time including child processes; \emph{wall-clock} is elapsed time on an otherwise idle machine; $n$ is the repetitions behind the row. The SCA oracle's spread is across environment seeds at one preference, so its range is printed rather than a standard deviation.
\item[$^{\ast}$] The teacher's search is paid once and returns its archive rather than one setting; the row times the corpus-source configuration.
\item[] Wall-clock is machine-dependent and not uniformly so---this machine runs the teacher's search faster than the evaluation server and the solvers' convex programs slower---while the interaction column does not move.
\end{tablenotes}
\end{threeparttable}
\end{table}

Two of the four routes pay in environment interaction and two pay in computation per slot (Table~\ref{tab:unifiedcost}). A rollout of the frozen model is one episode of 100 slots and $0.41$\,s. The solvers consume the same 100 slots and $252$ and $341$\,s, because every slot of that episode is a fresh convex program. The scalarized learner consumes 2M environment steps and $988$\,s before it can serve its first setting, and the second setting costs the same again, for a front $8.8$ HV below the swept one. The teacher's search consumes $606{,}000$ slots once and returns its archive---$60\times$ the interaction a 101-point sweep of the frozen model costs, a ratio between interaction counts and not between elapsed times. Per setting, on one core, the retraining route costs $2{,}400\times$ the rollout and the solvers $620$--$830\times$; a 21-point front follows by multiplication, at $8.6$\,s of rollouts against $5.8$ hours of retraining and two hours of re-solving for each evaluation seed. Densifying is asymmetric in the same way: $21\to101$ settings buys $+0.36$ HV for 80 further rollouts, while the teacher's front size is its population and changes only by searching again (Table~\ref{tab:incumbent}).

\begin{table}[t]
\centering
\caption{The scheduling node's per-slot cost, on the shipped frozen checkpoint and with no retraining.}
\label{tab:artifact}
\scriptsize
\setlength{\tabcolsep}{4pt}
\begin{threeparttable}
\begin{tabular}{@{}lr@{}}
\toprule
Quantity & Shipped checkpoint\\
\midrule
Parameters & 3.23\,M\\
Size (half precision) & 6.5\,MB\\
Decision latency, p95 & 3.7\,ms\\
\quad as a share of a 1\,s slot & 0.37\%\\
Hardware & one Apple M1 CPU core, batch 1\\
Control channel vs.\ offloaded traffic & $1.2\times10^{-3}$\\
\bottomrule
\end{tabular}
\begin{tablenotes}[flushleft]\scriptsize
\item[] The control-channel row compares the state the node receives and the actions it returns against the offloaded traffic the fleet already carries; the closed-form features of Sec.~\ref{sec:token} add nothing to it, being functions of raw state the node computes for itself.
\end{tablenotes}
\end{threeparttable}

\end{table}

The asymmetry does not depend on which route reaches the better front. It rests on where the search sits. The teacher must interact with the environment for all $606{,}000$ slots and the scalarized learner for 2M steps per setting, while PrefDT trains from stored trajectories---zero environment steps, $40{,}000$ gradient steps---and pays the search cost once, before any setting is requested (the axes this does not cover---the teacher's own quality, and the designer who knows which ten parameters to write---are recorded in Sec.~\ref{sec:conclusion}). What remains at run time is one forward pass per slot, priced in Table~\ref{tab:artifact}: nothing is solved, searched, or explored in the field. What offline training removes is the search at run time; what it cannot remove is the need for competent behavior to exist somewhere in the corpus.

\subsection{Attribution by ablation: corpus, encoder, return channel, pipeline}
\label{sec:attribution}
\label{sec:corpusadv}
\label{sec:iface}
\label{sec:ladder}
\label{sec:mechanism}
\label{sec:corpus}
Sec.~\ref{sec:prefdt} argued for each component of PrefDT before any experiment was run. Eleven of them are removed here, one at a time, with architecture, training schedule and seeds otherwise fixed, so that what leaves with each component can be attributed to it rather than to model capacity, to data volume, or to some other component. Table~\ref{tab:ablation} is the result, module by module, with the configuration each removal was run under.

\begin{table*}[t]
\centering
\caption{Every component of Sec.~\ref{sec:prefdt}, the ablation that removes it, what the ablation returns, and what that establishes. Arms, corpora and seed counts are in Appendix~C of the supplementary material.}
\label{tab:ablation}
\scriptsize
\setlength{\tabcolsep}{3.2pt}
\begin{tabular}{@{}p{2.4cm}p{4.2cm}p{4.9cm}p{4.2cm}@{}}
\toprule
Module & How it is removed & Result & What it establishes\\
\midrule
Attention pooling & zero padding of a fixed user count & $-2.90$ HV ($p{=}2.0{\times}10^{-7}$) & the invariant summary is the better representation, not merely the convenient one\\
\addlinespace[2pt]
Per-user bypass \eqref{eq:bypass1} & decision head reads the pooled summary only & $-6.10$ HV ($p{=}3.0{\times}10^{-29}$); $|\rho|$ $0.91\to0.68$ and reach $540\to491$\,s, both clauses of \eqref{eq:criterion} failing & the collapse \eqref{eq:perminv} predicts, produced on demand: a pooled encoder cannot supervise per-entity decisions without a per-entity path\\
\addlinespace[2pt]
Injection site \eqref{eq:inject} & $\boldsymbol{\omega}$ into the state token only, or into the return token only & $92.17/0.92/464$ and $90.56/0.80/552$ against the two-site model's $89.91/0.90/695$ (HV\,/\,$|\rho|$\,/\,reach) & two injection points buy reach, not front quality\\
\addlinespace[2pt]
Vector return-to-go \eqref{eq:inject} & scalar return, or no return at all & scalar $+0.08\%$ HV ($p{=}0.955$) at $|\rho|{=}0.69$, vector $+1.10\%$ ($p{=}1.1{\times}10^{-4}$) at $0.91$, both against no return at all; under the scalar, realized consumption reverses as the request rises, $54.1\to51.4\to45.2$\,kJ at requests of $70$, $80$ and $90$\,kJ & a scalar carries the setting and cannot be trained from\\
\addlinespace[2pt]
Preference token & removed, with the conditioner's targets supplied directly & HV $92.02$ against the un-ablated $90.28$; $|\rho|{=}0.805$ and reach $274$\,s, both below criterion & the token buys the settable front, not front quality\\
\addlinespace[2pt]
Context window & $K_{\mathrm{ctx}}\in\{1,10,50\}$ crossed with the return token & $+3.64$ HV with the token and $-3.44$ without it; interaction $+7.08$ ($p{=}1.9{\times}10^{-24}$) & context tracks the decrementing target \eqref{eq:decr}; it is not a model of the world\\
\addlinespace[2pt]
Teacher and corpus & distil from per-preference PPO experts instead of the teacher's archive & $-8.42$ HV; the field reorders, with a corpus--learner interaction of $+4.87$ HV & imitation preserves the teacher's operating point where policy improvement leaves it\\
\addlinespace[2pt]
Conditioner \eqref{eq:reg} & linear preference-to-return fit; ideal-point targets & the quadratic fit moves $|\rho|$ from $0.809$ to $0.932$; A4 posts $|\rho|{=}0.98$ over a reach of $170$\,s & the conditioner decides whether the target a setting names is one the model can reach\\
\addlinespace[2pt]
Corner emphasis & uniform preference sampling & $|\rho|{=}0.895$, under the criterion's $0.90$ & reaching the front's low-delay corner costs obedience until it is repaired\\
\addlinespace[2pt]
Physics features \eqref{eq:physfeat} & raw tokenizer features only & $|\rho|$ restored, $0.895\to0.931$; without corner emphasis the same features buy corner depth instead & what the features buy depends on the corpus, and the pair carries the shipped configuration back over the criterion\\
\addlinespace[2pt]
Setting-axis recalibration & uniform-$\boldsymbol{\omega}$ sweep & 21 of 21 settings distinct and non-dominated against 12 of 21 spoiled; $|\rho|{=}1.000$ against $0.931$; front quality unmoved & recalibration changes what a setting names, not what the model can do\\
\bottomrule
\end{tabular}
\end{table*}

\begin{figure}[t]
\centering
\includegraphics[width=\linewidth]{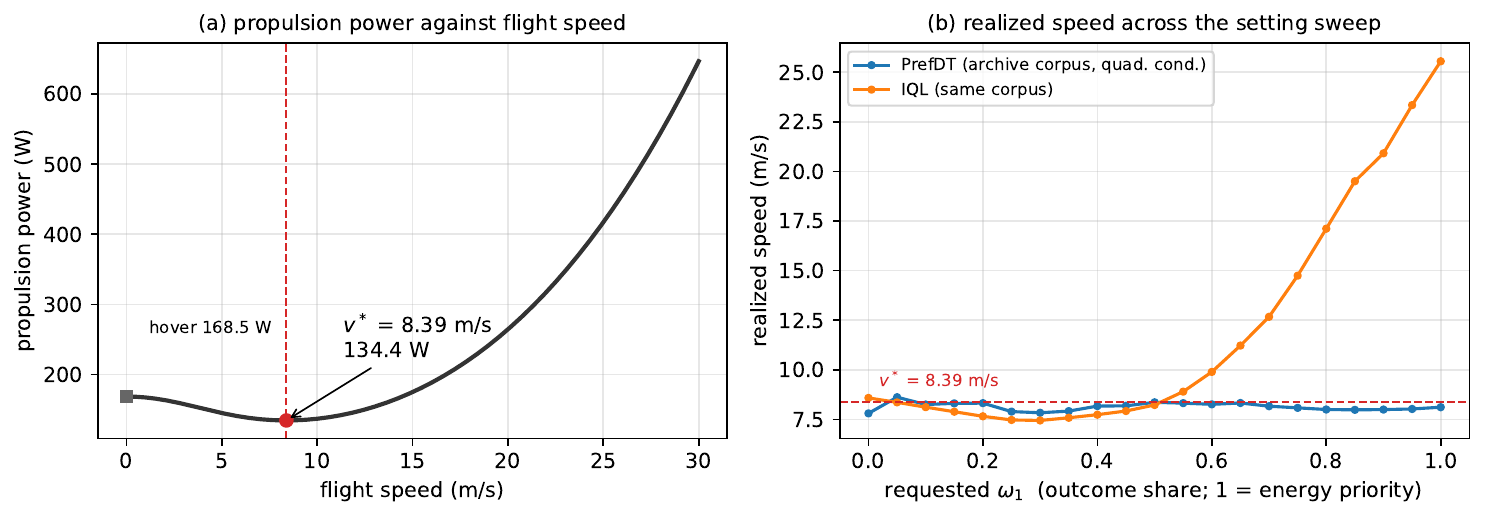}
\caption{(a) Propulsion power against flight speed, with hover and the interior minimum $v^*$ marked. (b) Realized mean flight speed against the requested setting, for two learners trained on the same archive corpus: PrefDT with the quadratic conditioner, and $\boldsymbol{\omega}$-conditioned IQL. The dashed line is $v^*$ carried over from (a).}
\label{fig:speed}
\end{figure}

Three readings run across the table. The per-user bypass is not a choice between alternatives: the loss its removal produces is the one \eqref{eq:perminv} predicts---front quality, fidelity and reach fall together---so what the measurement establishes is that a pooled encoder cannot supervise per-entity decisions without a per-entity path. The pipeline that follows the corpus is a short ladder: the archive-corpus model reads $98.65$, the quadratic conditioner carries it to $99.04$, corner emphasis buys the front's corner at a fidelity of $0.895$, and the physics features restore fidelity and land the shipped $99.47$. Among the choices that were genuinely open, the corpus moves front quality further than any architectural one---$8.42$ HV---because the teacher's search settled at the propulsion power minimum and imitation preserves that operating point where policy improvement disturbs it (Fig.~\ref{fig:speed}). And the removals fall into two kinds: those that cost front quality, and those that cost the settable front while leaving front quality alone or \emph{raising} it---removing the preference token raises hypervolume above the un-ablated configuration with both clauses of \eqref{eq:criterion} failing. Selecting by front quality alone, in this design space, selects a model that does not take settings.

\subsection{Deployment: what the link need not guarantee, and what scale does not widen}
\label{sec:uplink}
\label{sec:deployment}
Execution is centralized: one node---a ground station, or a designated leader aircraft---receives fleet state once per slot and returns the joint action. That link loses reports, and it takes time. This section measures what the scheduler needs it to guarantee, and reports three requirements it does not impose; what it does impose is Sec.~\ref{sec:boundaries}.

The first requirement is delivery. We drop a fraction of the user reports before they reach the policy---$5$, $10$, $20$ and $40\%$---and read what fraction of the front survives against the same model's own nominal condition; the registered predictions this study made, and the one that failed, are in Appendix~C. At one report in twenty, $96.7\%$ of the front survives; at two in five, with no retraining and no special handling, $71.1\%$ does, where the scripted per-user rule run on identical loss masks keeps $60.8\%$. The scheduler therefore keeps deciding on a link that drops reports, and the reason is the pooled summary of Sec.~\ref{sec:pooling}: it is defined on any subset of the active set, so a lost report delivers a smaller set rather than an incomplete input, and no substitute value is invented for the members that are missing. Nothing in training saw a dropped report.

The second requirement is recency of the individual report. Substituting the last report received for a missing one lifts retention at every loss rate---by $1.0$, $2.1$, $4.3$ and $8.9$ points at $5$, $10$, $20$ and $40\%$, so that at two fifths of all reports lost $80.0\%$ of the front still survives---while the same substitution applied to the per-user rule \emph{costs} it $1.3$, $2.3$, $3.7$ and $3.3$ points at the same four rates, each direction significant on the same 100 paired seeds. A stale value entering the pooled summary is one member among those that did arrive and is diluted accordingly, while a rule that compares this user's own numbers against a threshold acts on the stale value directly. The obvious mitigation is therefore available to this scheduler and unavailable to the alternative, and one buffered report per user recovers most of what heavy loss costs.

The third requirement is speed. We displace the snapshot the policy acts on by $k$ slots and read retention in both directions, so that staleness is separated from information: a snapshot from one slot in the \emph{future} damages exactly as much as one from a slot in the past (paired difference $0.007$ percentage points, $p{=}0.85$), and retention tracks the overlap between the snapshot's activity set and the current one, about $0.43$ at every displacement in both directions. What the scheduler requires is therefore that the snapshot be the current slot's, not that it arrive quickly: within a slot the active set cannot change, by the discrete-slot formulation of Sec.~\ref{sec:model}, so any transport completing inside the slot costs nothing and no latency guarantee is needed. What it costs when the snapshot is \emph{not} the current slot's is Sec.~\ref{sec:boundaries}.

A fourth condition is the size of the fleet. The encoder's parameter shapes bind $M$, so a different fleet is served by retraining at that size rather than by one model serving every size, and what that retraining preserves is the ratio to the teacher: $78$--$80\%$ at $M{=}2/4/8$. The distillation gap therefore does not widen with scale, and the per-slot decision stays under $2\%$ of the slot at every size.

\subsection{The boundaries: state synchronization, and the trained operating point}
\label{sec:boundaries}
Sec.~\ref{sec:uplink} established that what the link must deliver is a snapshot of the current slot, not a fast one. This experiment measures what happens when it does not: the policy decides slot $t$ from the state of slot $t{-}k$, for $k$ up to five. One slot of displacement costs about half the front, retention $49.5\%$, and further displacement costs no more---between $49.0$ and $50.0\%$ out to five slots. The loss is a cliff at the first slot rather than a slope, so there is no operating point at which a late snapshot is partially acceptable. State synchronization is a requirement of this scheduler, and it is the first boundary of the envelope.

The second boundary is the operating point itself. Every result above is measured at the configuration the model was trained on, while a deployed fleet meets channel constants, task loads and propulsion parameters that differ from it, so what this experiment asks is how far those results carry. Re-evaluated frozen at fifteen operating points away from the training configuration, PrefDT's worst-case front quality drops $4.3\%$ on average, against $2.6\%$ for the searched rule its corpus came from (Appendix~D): imitation inherits the rule's behavior at the trained point and not its tolerance away from it. Zero-shot transfer to shifted physics is therefore outside what this paper claims. An in-episode switch of the arrival process, by contrast, degrades neither at any context length, which locates the budget contract of Sec.~\ref{sec:budget} precisely: what it survives is a disturbance to the \emph{integral} of consumption, the quantity its token carries, not a change in the environment's statistics.

These two boundaries bound the envelope. Inside it---at the trained operating point, with the snapshot of the current slot---reports may fail and transport may take milliseconds, and every result of Sec.~\ref{sec:operability} applies as stated.

\section{Conclusion}
\label{sec:conclusion}
Unlike existing UAV-MEC schedulers, which compile the energy--delay trade-off into the objective before the fleet flies, PrefDT takes it as an input the scheduler reads at inference. One frozen model then serves any point on the curve at one rollout, $0.41$\,s per setting, where the alternatives pay a fresh solve, a fresh search, or a fresh training run.

The same channel carries a physical contract. Its energy entry is the joules still allowed, so a total budget is honored against what the fleet actually spends, is revisable in flight by a single write, and saturates safely when more is asked for than can be delivered. A denser front costs further rollouts and no further search. Three components make this work in a physical network: a pooled-token encoder with a per-user decision bypass, a vector-valued conditioning channel carried by a rollout decrement, and a seven-step distillation pipeline that generates the preference-labeled corpus that scheduling does not supply.

Across 26 method variants under one protocol, PrefDT produces the best trade-off curve of any learned method and, on the corpus it ships with, is the only one that clears a criterion on setting--outcome order and span fixed before the decisive runs. One model, unchanged, also holds an energy budget to $0.6\%$ when the physics change in mid-flight, where the open-loop rule it learned from overruns on every episode; and its front reaches $99.7\%$ of a teacher stronger than the one it was distilled from. Front quality does not say whether a setting was followed. We report two readings that do---the order in which the outcomes arrive, and the span over which that order holds---and the configuration this paper ships was selected on them, not on front quality.

Four qualifications bound what this paper establishes. Every result is measured in one physics-based simulator family, whose i.i.d.\ arrival process the training corpora share. Uplink bandwidth is uncontended by construction (Sec.~\ref{sec:model}), so one source of pressure on the association decision is absent. The solver latency of Sec.~\ref{sec:switchcost} times a per-slot re-compilation that parameterized caching would largely amortize---only $79$\,ms per slot is irreducible---so the schedulability comparison is a statement about our implementation, not about the methods; what no caching removes is the structure the cost table prices, a fresh convex program every slot and a fresh solve per preference. Finally, the cost accounting does not cover design knowledge at all: the teacher's quality is ours to have composed, and a designer who already knows which ten parameters to write has paid a cost this paper does not price, incurred once when the system is built and never again at a setting.

\begin{IEEEbiography}[{\includegraphics[width=1in,height=1.25in,clip,keepaspectratio]{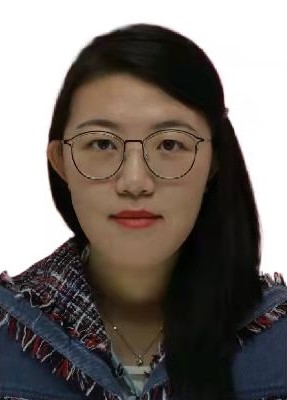}}]{Qiao Liao}
is currently pursuing a Ph.D. degree in Computer Science and Technology with the College of Intelligence and Computing, Tianjin University.
She is working on mobile edge computing. Her main research interests are reinforcement learning and service computing. 
\end{IEEEbiography}
\vspace{-10mm}
\begin{IEEEbiography}[{\includegraphics[width=1in,height=1.25in,clip,keepaspectratio]{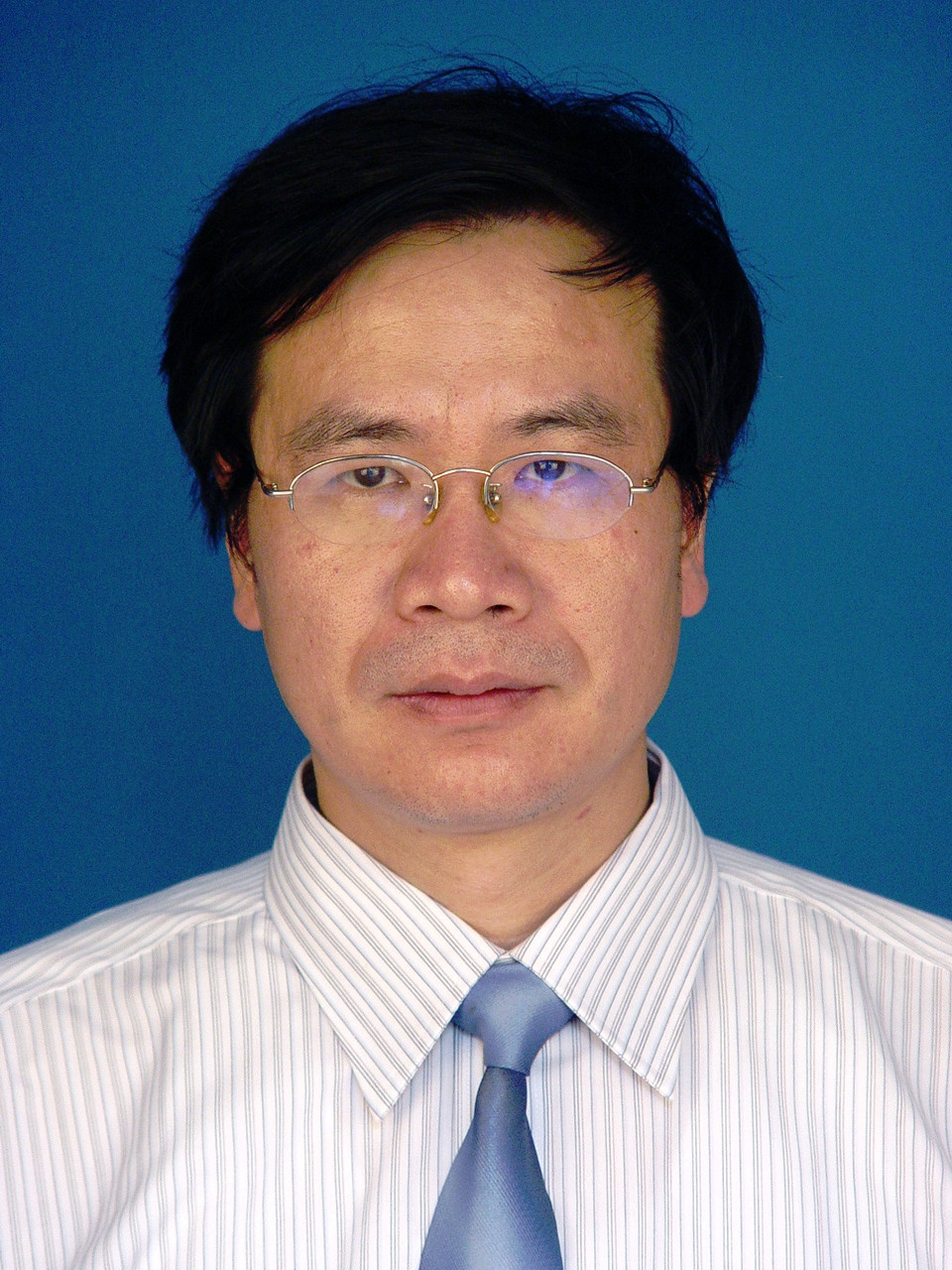}}]{Zhiyong Feng}
(Member, IEEE)  was born in 1965.
He received the Ph.D. degree from Tianjin University, Tianjin, China, in 1996.

He is currently a Professor with the College of Intelligence and Computing, Tianjin University, Tianjin, China.
He has authored more than 200 articles, one book and 39 patents.
His research interests include service computing, knowledge engineering, and software engineering.

Dr. Feng is a distinguished member of China Computer Federation (CCF),
a member of the Association for Computing Machinery (ACM),
and the Chairman of ACM China Tianjin Branch.

\end{IEEEbiography}

\vspace{-12mm}

\begin{IEEEbiography}[{\includegraphics[width=1in,height=1.25in,clip,keepaspectratio]{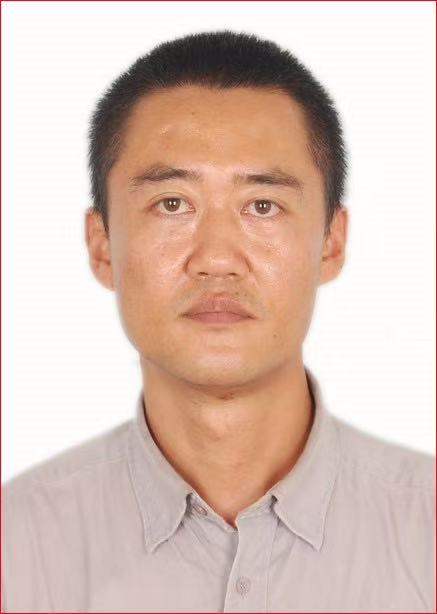}}]{Bin Wu}
received the Ph.D. degree in Electrical and Electronic Engineering from the University of Hong Kong, Hong Kong, in 2007. From 2007 to 2012, he was a Postdoctoral Research Fellow with the Department of Electrical and Computer Engineering, University of Waterloo, Waterloo, ON, Canada. He is currently a Professor with the College of Intelligence and Computing, Tianjin University, Tianjin, China. His research interests include computer systems and networking, and communication system design.
\end{IEEEbiography}

\vspace{-10mm}

\begin{IEEEbiography}[{\includegraphics[width=1in,height=1.25in,clip,keepaspectratio]{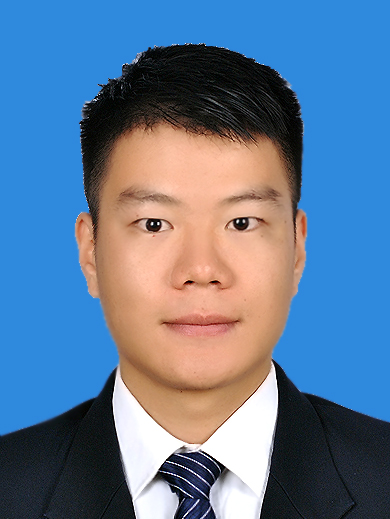}}]{Guodong Fan}
received the Ph.D. degree in Computer Science and Technology with the College of Intelligence and Computing, Tianjin University, Tianjin, China. He is working on cognitive services. His main research interests include representation learning, service computing, and software repository mining.
\end{IEEEbiography}

\end{document}